\documentclass[journal]{IEEEtran}
\IEEEoverridecommandlockouts
\makeatletter
\def\endthebibliography{%
  \def\@noitemerr{\@latex@warning{Empty `thebibliography' environment}}%
  \endlist
}
\makeatother
\usepackage{cite}
\usepackage{amsmath,amssymb,amsfonts}
\usepackage{amsthm}
\usepackage{algorithmic}
\usepackage[ruled]{algorithm2e}
\usepackage{bm}
\usepackage{setspace}
\usepackage{indentfirst}
\usepackage{graphicx}
\usepackage{stfloats}
\usepackage{epstopdf}
\usepackage{textcomp}
\usepackage{xcolor}
\usepackage{tikz}
\usetikzlibrary{arrows,shapes, chains}
\usepackage[caption=false]{subfig}
\usepackage{multirow}
\usepackage{balance}

\usepackage[colorlinks,
            linkcolor=black,
            anchorcolor=black,
            citecolor=black]{hyperref}
\usepackage{url}
\usepackage{mathrsfs}
\usepackage{verbatim}

\newtheorem{lemma}{Lemma}

\ifodd 1
\newcommand{\com}[1]{{\color{black}#1}} %comment of the text
\else

\newcommand{\com}[1]{}
\fi

\begin{document}
\title{{ Data-Driven Online Resource Allocation for User Experience Improvement in Mobile Edge Clouds}}
\author{\IEEEauthorblockN{Liqun Fu, ~\IEEEmembership{Senior Member,~IEEE},
Jingwen Tong,  ~\IEEEmembership{Member,~IEEE},
Tongtong Lin,
and
Jun Zhang,  ~\IEEEmembership{Fellow,~IEEE}
}
\thanks{Corresponding author: Liqun Fu.
This work was presented in part at IEEE ICC 2021 \cite{lin2021online}.
L. Fu and T. Lin are with the  School of Informatics, Xiamen University, Xiamen 361005, China
(e-mails: liqun@xmu.edu.cn; tongtong@stu.xmu.edu.cn).
J. Tong and J. Zhang are with the Department of Electronic and Computer Engineering, The Hong Kong University of Science and Technology, Hong Kong, China (e-mails: eejwentong@ust.hk; eejzhang@ust.hk).
}}
\maketitle
\begin{abstract}
As the cloud is pushed to the edge of the network, resource allocation for user experience improvement in mobile edge clouds (MEC) is increasingly important and faces multiple challenges.
This paper studies quality of experience (QoE)-oriented resource allocation in MEC while considering user diversity,  limited resources, and the complex relationship between allocated resources and user experience.
We introduce a closed-loop online resource allocation (CORA) framework to tackle this problem. It learns the objective function of resource allocation from the historical dataset and updates the learned model using the online testing results.
Due to the learned objective model is typically non-convex and challenging to solve in real-time,
we leverage the Lyapunov optimization to decouple the long-term average constraint and apply the prime-dual method to solve this decoupled resource allocation problem.
Thereafter, we put forth a data-driven optimal online queue resource allocation (OOQRA) algorithm and a data-driven robust OQRA (ROQRA) algorithm for homogenous and heterogeneous user cases, respectively.
Moreover, we provide a rigorous convergence analysis for the OOQRA algorithm.
We conduct extensive experiments to evaluate the proposed algorithms using the synthesis and YouTube datasets.
Numerical results validate the theoretical analysis and demonstrate that the user complaint rate is reduced by up to $100 \%$ and $18 \%$  in the synthesis and YouTube datasets, respectively.
\end{abstract}


\section{Introduction}\label{I}
With the rapid development of wireless communications, there is a huge increase in the data generated in future networks;
While the advanced storage and computing technologies strongly demand real-time data processing \cite{zhang2019mobile, siriwardhana2021survey}.
Motivated by this trend, mobile edge cloud (MEC), which is also known as mobile edge computing,  has emerged as an effective technology to alleviate these issues by placing multiple edge servers proximate to the base station (BS) \cite{mao2017survey}.
The benefits of the MEC system are twofold: i)  improve bandwidth utilization and energy efficiency for communications;
ii) reduce transmission delay and the demand for backhaul links.
In recent years, MEC has attracted increasing attention in wireless communications, predicted to increase by $26.4 \%$ annually from 2022 to 2026 \cite{Mordor2022}.

However, as the cloud is pushed to the edge of the network, resource allocation for user experience improvement in the MEC system is increasingly important and facing multiple challenges.
Many efforts have been devoted to user-oriented resource allocations, which also refers to the quality of experience (QoE)-based resource allocation. The QoE is a critical performance metric in wireless communications \cite{liu2017mobile} and is defined as the outcome of the  joint considerations of multiple quality of services, such as good throughput, low delay, and high energy efficiency,
Typically, a poor QoE will undermine user experience, resulting in severe profit loss for the MEC operators.
Accounting for this, many works focus on the QoE-oriented resource allocations in MEC with different considerations of the QoE-based video streaming problem in \cite{li2017qoe11, ge2018qoe13, tuysuz2020qoe},  the QoE-based task or computation offloading problem in \cite{he2020qoe16, luo2019qoe33, sivasakthi2022qoe34}, and the QoE-based joint service placement and throughput adjustment problem in \cite{liang2022online}.
The goal of these works is to improve the users' QoE by enabling users to utilize the cloud resource without any requirement of complex hardware and software implementation.

Unfortunately, several practical factors prevent the above QoE-based resource allocations from achieving a desirable performance \cite{jiang2021intelligence, shang2022enabling}.
First, users are typically heterogeneous, where different users may have different resource requirements (or features).
For example, a video streaming device may require high bandwidth and computation capability to reduce the transmission and processing delay;
while an Internet of Things (IoT) device prefers a large memory space to store its sampling data.
Moreover, the QoE may change with time for the same type of resources for those homogeneous users.
For instance, a mobile user typically needs less bandwidth for social media in the daytime but desires a large bandwidth for streaming video in the nighttime.
These features render offline methods inapplicable to QoE-based resource allocations.
An efficient online method is urgent to handle such changing features and user diversity in the MEC system.

Second, the relationship between user experience and resource features is complex.
Existing works usually assume that this causal relationship is well determined and can be formulated as some analytic utility functions \cite{DD2017}.
However, such an assumption cannot directly apply to the MEC system, especially when the environment is dynamic.
On the one hand, user movement often leads to a dynamic objective function and a long delay of service.
For example, users may move from one coverage area of the BS to another.
On the other hand, the MEC systems often involve dynamic and time-varying workloads,
resulting in poor user QoE and a high system burden.
Therefore, an edge server must allocate resources without this knowledge and change its allocation strategies in response to this dynamic feature.

Third, an edge server usually has limited resources constrained by physical limitations such as space, power, and cooling \cite{tuysuz2020qoe}.
The edge servers are typically smaller in size than traditional data centers or cloud servers.
In addition, the limited resources are typically formulated as a long-term constraint to handle the workload variability, fairness, and system stability in the MEC.
In other words,
the average allocated resources of all users should be constrained to a pre-determined value.
To sum up,  the relationship between user experience and resource features needs to be learned from the historical data and adjusted adaptively according to the allocated resources under limited resources.

To shed some light on this problem, we put forth a data-driven online resource allocation mechanism by considering the user diversity and limited resource constraints,
as well as the obscure causal relationship between user experience and allocated resources.
We first introduce a closed-loop online resource allocation (CORA) framework \cite{DG2016, DD2017} to  the homogeneous user case where users have the same sensitivity to the allocated resources.
By using the historical dataset, the CORA framework employs a logistic regression-based classifier \cite{ng2001discriminative} to learn the relationship between user experience and resource features.
Based on the learned model,  it constructs an objective function for this online resource allocation problem.
By solving this problem, the users evaluate the assigned resource.
In return, the evaluation results (i.e., user experiences) are fed back to the classifier to update the learned model (i.e., the causal relationship).

There are two challenges in solving this data-driven resource allocation problem.
First, the limited resource constraint usually appears in a long-term average form, which cannot be directly applied to the online learning method.
To overcome this difficulty, we transform it into a one-slot optimization problem (i.e., the online learning feature) by using the Lyapunov optimization technique \cite{mao2015lyapunov}, which maintains a virtual queue to track the violation of this constraint.
Second, the transformed problem is still computationally intractable due to its non-convex feature.
We propose a data-driven optimal online queue resource allocation (OOQRA) algorithm to solve the transformed problem by using the prime-dual method.
The OOQRA algorithm proceeds sequentially and can efficiently capture the network dynamics.
Moreover, we show that the OOQRA algorithm is guaranteed to converge to the optimal allocation by giving a rigorous convergence analysis.
To sum up, the data-driven online resource allocation mechanism is operated in a closed loop and has the following features:
1) It captures the user's mobility and the time-changing resource requirement;
2) It extracts the valuable model from the historical dataset and dynamically adjusts it through user experiences on the allocated resources.

Furthermore, heterogeneous users will introduce congestion in resource allocation as they may be more sensitive to the allocated resources.
To overcome this challenge, we incorporate a coefficient component to the above CORA framework to estimate the average behavior of this relationship by using the multi-armed bandit (MAB) technique \cite{bubeck2012regret, tong2022age, chen2022over}.
In this MAB problem, the player is the edge server, and the arms are the elements of the coefficient matrix, which describes the relationship between feature value and allocated resources.
On the one hand, the player needs to explore each arm as many times as possible to find the optimal solution;
on the other hand, it requires exploiting the current best arm as many times as possible to reduce its performance loss.
Therefore, we leverage the well-known upper confidential bound (UCB) algorithm  \cite{auer2010ucb}  to learn this coefficient matrix by carefully balancing the above exploration and exploitation dilemma.
Then, we propose a data-driven robust OQRA (ROQRA) algorithm for heterogeneous users based on the improved CORA framework.

The main contributions of this work are summarized below.
\begin{itemize}
	\item We investigate the  QoE-based resource allocation problem in MEC while considering user diversity, limited resources, and the complex causal relationship between user experience and allocated resources.
We introduce the CORA framework to tackle this problem by making full use of the historical dataset and user experiences.
	\item For the case of homogenous users, we propose an OOQRA algorithm to solve this problem by exploiting the prime-dual method and the properties of optimal solutions. In addition, we provide a rigorous convergence analysis for the OOQRA algorithm.
 \item For the case of heterogeneous users,  we first design a coefficient component in the CORA framework to estimate the coefficient matrix by using the UCB algorithm.
 Then, we propose a ROQRA algorithm to solve this resource allocation problem.
   \item Finally, we evaluate the proposed algorithms by conducting extensive experiments on the synthesis dataset and the YouTube dataset.
   The numerical results show that the user complaints rate can be reduced by up to $100\%$ and $18\%$  when the available resources increase for the synthesis and the YouTube datasets, respectively.
\end{itemize}

The rest of this paper is organized as follows.
In Section \ref{II}, we introduce the system model of the QoE-based resource allocation in the MEC system.
Section \ref{III} gives the CORA framework to model the QoE-based resource allocation problem.
The proposed algorithms to solve the above resource allocation problem are given in Section \ref{IV} for the homogenous users and in Section \ref{V} for the heterogeneous users.
The numerical results are presented in Section \ref{VI}, and this paper concludes in Section \ref{VII}.

\begin{figure}[!t]
	\centering
	\includegraphics[width=0.80\columnwidth]{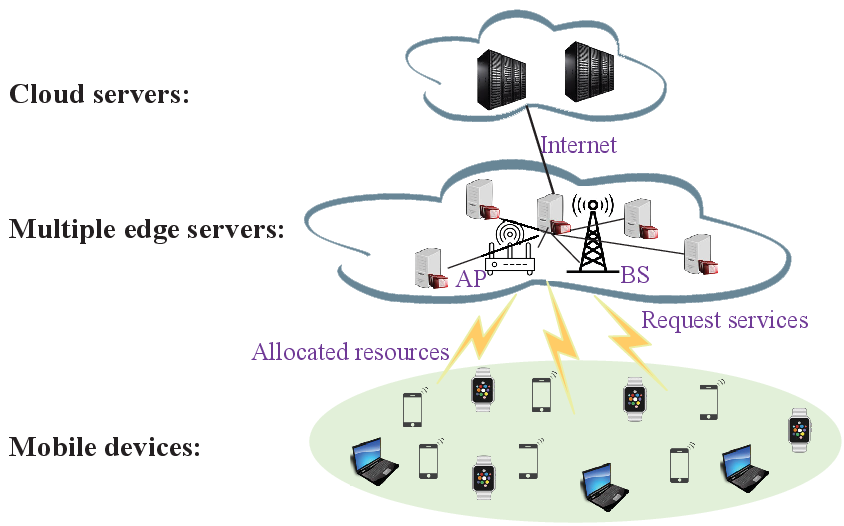}
	\caption{An illustration of the mobile edge could system.}
	\label{cloud}
\end{figure}
\section{System Model}\label{II}
We consider an MEC system, as shown in Fig. \ref{cloud},  which consists of three layers:  mobile devices (or end users), multiple edge servers, and cloud servers.
The mobile devices are connected to a BS or AP, which usually refers to IoT devices or smartphones.
These devices may send requests to the BS for services,
such as  IoT analytics, video streaming, and cloud gaming.
To improve transmission efficiency and reduce computation latency,
the BS will offload service to the edge servers deployed close to the BS
rather than forward them to the remote cloud server.
After receiving the request,
the edge server executes this task by allocating suitable resources to the user.
For example, the allocated resources can be the bandwidth, computation, and storage capability, \textit{etc}.
Note that a cloud server usually deploys at the edge of the mobile network to further extend the computing capabilities, as shown in Fig. \ref{cloud}.

We assume that time is slotted in $t=1, 2,\ldots, T$.
At each time slot $t$, there are $S_t=\{1,2,\ldots \}$ users requesting service from the edge servers as users' requests are generated randomly in the cellular network and some computing tasks are generated periodically, such as the IoT analytic task.
This indicates that at least one user requests service at each time slot.
According to the request service, the edge servers will allocate suitable resources to the user.
The allocated resources may change the user's features and then the user experience.
In this work, we assume that the user experience with respect to the allocated resources is labeled as negative or positive\footnote{
In fact, this binary classification model has widely existed in practical applications. For instance, some MEC operators prefer to conduct a binary service satisfaction questionnaire to the users\cite{DD2017}.}.
If the user is not satisfied with this service, its QoE is labeled as a complaint denoted by `1' (positive); otherwise, its QoE is marked as `0' (negative).
\com{Note that we consider the real-time feedback setting here, where the server can receive the user experience immediately after resource allocation. However, for the asynchronous feedback, one can design an implicit feedback mechanism to produce the instant feedbacks, which also collects in real-time but without explicit user experience input. Then, the implicit-explicit hybrid feedback can be used to guide resource allocation.}

Mathematically, assume that each user has a $D$-dimensional feature vector denoted by ${{\mathbf{x}}}(t) = [ {x_{1}(t)},{x_{2}(t)},\ldots,$ ${x_{D}(t)} ]^\mathsf{T}$, where ${x_{d}(t)} $ is  the $d$-th feature value of the $t$-th user.
Note that each feature can be regarded as a specific performance metric, such as throughput, delay, or energy efficiency.
In addition, the edge server has $K$ types of available resources, such as bandwidth, computation, and storage capability.
The bandwidth resource may influence the throughput and delay performance, resulting in changed user feature values.
We adopt a linear model\footnote{\com{The linear model \cite{DD2017} is considered as our first step in the data-driven resource allocation problem. This simple model captures the features we consider here, such as throughput. We will explore more complex and general functions in future.}} to quantify the relationship between the feature value and the allocated resources.
Then, the user's feature vector $\mathbf{x}(t)$ after resource allocation is updated by
\begin{equation}\label{UpFeaVal}
{\mathbf{g}}(t)={\mathbf{x}}(t) + { \mathbf{Z}}(t){{\mathbf{r}}}(t),
\end{equation}
where ${\mathbf{g}(t)} \in \mathbb R^{D \times 1}$ is the updated feature vector.
Matrix ${\mathbf Z}(t) \in \mathbb R^{D \times K}$ represents the resource coefficient of the $t$-th user,
and element $z_{d,k}(t)$ denotes the effect of the $k$-th resource on the $d$-th user feature.
For convenience, the main notations used in this paper are summarized in Table \ref{tab_sym}.

In the following, we consider two types of network scenarios, i.e., the users are homogeneous or heterogeneous.
For the case of homogeneous users, all users have the same resource coefficient matrix ${\mathbf Z}$.
For the case of heterogeneous users, different users have different resource coefficient matrix ${\mathbf Z(t)}$.
In addition, we consider that only one user requests a service from the edge server at each time slot.
The system's goal is to maximize the long-term average QoE or reduce the total complaint rate by sequentially allocating limited resources to these homogeneous or heterogeneous users.

\begin{table}[!t]
\caption{List of Main Notations}
\centering
\begin{tabular}{l|l}
\hline
Notation & Description \\ [0.5ex]
\hline\hline                 % inserts single-line
$D$   &  Total number of user features \\  % Entering row contents
$K$   & Total types of the available resources \\
 $I$   & Size of the initial dataset   \\
 $T$   & Size of the online dataset    \\
 $B_k$   & Upper bound of resource $k$   \\
 $x_d(t)$   & The value of $d$-th user feature at slot $t$    \\
 $r_k(t)$   & The value of $k$-th allocated resource at slot $t$    \\
 $z_{d,k}(t)$   & The effect of $k$-th resource on $d$-th user \\
                         & feature at time slot $t$    \\
 $L_{d,k} (t)$  & The number of unit resources in the $(d, k)$-th \\
                          & arm at time slot $t$    \\
 $Q_k(t)$   & The virtual queue of resource $k$ at time slot $t$    \\
 $v_k, \tau_k$  & The Lagrangian multiples of the $k$-th resource    \\
 $\mathbf{w}$   & The weight vector of the classification model    \\
 $\hat {\mathbf \Psi}$   & The resource coefficient matrix    \\
 \hline                        % inserts single-line
 \end{tabular} \label{tab_sym}
 \end{table}

\section{Problem Statement} \label{III}
We first introduce the CORA framework in Fig. \ref{first_model} to handle the  QoE maximum problem.
This framework mainly consists of a classification component,  a resource allocation component, and an evaluation component.
As shown in Fig. \ref{first_model}, these components are operated in a closed loop.
The classification component is used to learn the causal relationship between the user feature value and the user experience by using the logistic regression-based classifier.
The resource allocation component adopts the learned model to assign resources to new users.
After allocation, the user evaluates the allocated resources, where the testing result is regarded as the ground truth (i.e., the binary service satisfaction questionnaire) to fed back the classification component for model learning.
In the following, we illustrate the CORA framework in detail and formulate the resource allocation problem as a stochastic optimization problem.
 \begin{figure}[!t]
	\centering
	\includegraphics[width=0.85\columnwidth]{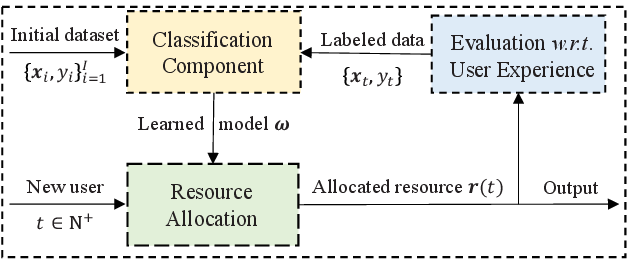}
	\caption{The closed-loop online resource allocation (CORA) framework.}
	\label{first_model}
\end{figure}

\subsection{ Logistic Regression-Based Classification Component}
We adopt the logistic regression model\footnote{According to \cite{DD2017}, the logistic regression model has the best performance among existing classifiers for predicting the causal relationship.} to the classification component to learn the objective model based on the historical and online dataset.
This dataset contains the ground truth information (or relationship) of the user experience and the allocated resources,
which is labeled as negative or positive.
Let $\Omega_0 =  \{\mathbf x_i, y_i\}^{I}_{i=1}$ be the  initial dataset,  where $I$ denotes the total number of existing users.
Note that $I$ can be any nonnegative integer value in the homogenous user case and a fine-tuning parameter in the heterogeneous user case.
In addition,  $\mathbf x_i \in \mathbb R^{D\times 1}$ and $y_i \in\{0,1\}$ are the feature and label of user $i$, respectively.
Let $T$ be the total number of time slots.
Hence, the overall data size is denoted by $\Omega_{T}= \{\mathbf x_i, y_i\}^{I}_{i=1}\cup \{\sum_{s=1}^{S_t} (\mathbf g_s(t), y_s(t))\}^{T}_{t=1}$, where $\mathbf g_s(t)$ and $y_s(t)$ are the updated feature vector and label of  user $s$ at the time slot $t$, respectively.

Unlike conventional classifiers, the goal of the classification component in Fig. \ref{first_model} is not to maximize the accuracy of the predicting user labels solely,
but to better quantify the relationship between user experience and resource allocation outcomes.
Mathematically, the objective of the classification model is to minimize the cross-entropy cost function, i.e.,
\begin{equation}\small
\begin{aligned}
&{\mathscr{P}_{\textrm{0}}}:\
\mathop {\max }\limits_{w_0,\mathbf w} \mathop {\lim}\limits_{T\to \infty }
J_T(w_{0}, \mathbf w)\\
& = \frac{1}{N_T}  \left\{\sum_{t= 1}^{T}\sum_{s= 1}^{S_t} \Bigg({y}_s(t)\mathop{\log}\bigg(\frac{1}{1 + \mathbb{\exp}( - {w_{0}} - \mathbf w^\mathsf{T}{{\mathbf{g}_s}(t)})}\bigg)\Bigg.\right.\\
&\Bigg.\left.+\left(1-y_s(t)\right)\mathop{\log}\bigg(\frac{1}{1 + \mathbb{\exp}( {w_{0}} + \mathbf w^\mathsf{T}{{\mathbf{g}_s}(t)})}\bigg) \Bigg)+\mathbf \Gamma( {w_{0}}, \mathbf w)\right\},\label{logiticregression}
\end{aligned}
\end{equation}
where $\mathbf \Gamma( {w_{0}}, \mathbf w)$ is the loss function of the classifier using the initial dataset $\Omega_0$, i.e.,
\begin{equation}
\begin{aligned}
\mathbf \Gamma( {w_{0}}, \mathbf w)
&=\sum\limits_{i = 1}^{I} \Bigg\{ {y}_i\mathop{\log}\left(\frac{1}{1 + \mathbb{\exp}( - w_0 - \mathbf w^\mathsf{T}{{\mathbf{x}}_i})}\right) \\
&+(1-y_i)\mathop{\log}\left(\frac{1}{1 + \mathbb{\exp}( w_0+\mathbf w^\mathsf{T}{{\mathbf{x}}_i})}\right)\Bigg\},
\end{aligned}
\end{equation}
where  $N_T=I+\sum_{t=1}^{T}S_t$ is the total number of users in  dataset $\Omega_T$.
Note that there may be multiple users requesting service at each time slot, i.e., $S_t \geq 1$.
We can group the users who requested service at the same slot as a super user because the coefficient matrix ${\mathbf Z}$ is fixed in the homogenous user case.

The classifier intends to learn the parameters $w_0$ and $\mathbf w$, where $w_0$ represents the intercept and $\mathbf w=\left[w_1,w_2,\dots,w_D\right]^\mathsf{T}$ contains the weights related to features.
It can be seen that  ${\mathscr{P}_{\textrm{0}}}$ is a convex function,
which can be solved effectively using the average gradient descent (AGD)  algorithm \cite{boyd2004convex}.
Then, $\tilde{{\mathbf w }}_{t}=({w_{t,0}},{\mathbf{w}}_{t})$ is updated by
\begin{equation}
\tilde{{\mathbf w }}_{t} = \tilde{{\mathbf w }}_{t-1} -\frac{\eta_{t}}{N_{t}}\nabla J_t(\tilde{{\mathbf w }}_{t-1}),\label{weightsupdate}
\end{equation}
where $\eta_{t}$ is the step size of the $t$-th time slot.
Term  $N_{t}=I+t$ is the number of iterations and $\nabla J_t(\tilde{{\mathbf w }}_{t-1})$ denotes the gradient of $J_t(\tilde{{\mathbf w }}_{t-1})$.

\subsection{Data-Driven Resource Allocation Component}
For the ease of notation, we define $h(x)  \triangleq ({1 +\exp(x)})^{-1}$.
Thus,  the positive probability of the $s$-th user at  time slot $t$ is computed by $h\left(-w_{{t-1},0}-\mathbf w^\mathsf{T}_{t-1}\mathbf x_s\left(t\right)\right)$.
After resource allocation, the positive probability is predicted by
\begin{equation}\label{pos_rate}
u[{{ \mathbf{r}_s}}(t)]=h\left(C_s(t)+{\mathbf a}_s(t) \mathbf{r}_s(t)\right),
\end{equation}
where  $C_s(t)=-{{\mathbf{w}^\mathsf{T}_{t-1}}}{\mathbf{x}_s}(t) -{w_{t-1,0}}$ and $\mathbf{r}_s(t)$ is the allocated resources at the $(s,t)$-th user.
In addition, ${\mathbf a}_s(t)=-{\mathbf Z}_s(t){^\mathsf{T}}{\mathbf{w}}_{t-1}$
and  ${\mathbf{a}_s(t)} = \left[ {{a^s_1}(t),{a^s_2}(t),...,{a^s_K}(t)} \right]^\mathsf{T}$ is the resource efficiency vector, where $a^s_k(t)$ represents the efficiency of the resource $k$ on user $t$.
When $a^s_k(t) \le 0$, the allocation of resource $k$ impairs the performance of system.
Hence, we assume that $a^s_k(t) > 0$.

The objective of the resource allocation component is to minimize the average complaint rate, which can be formulated as the following stochastic optimization form, i.e.,
\begin{align}
{\mathscr{P}_{\textrm{1}}}:
\mathop {\textrm{minimize}}\limits_{\mathbf r_s(1),\dots,\mathbf r_s(T)} \quad &\mathop {\lim}\limits_{T \to \infty }\frac{1}{T}\sum_{t = 1}^{T}\sum_{s = 1}^{S_t}  u[{{\mathbf r}_s}(t)]\label{p1}\\
\textrm{subject\ to}{\quad}&\mathop {\lim}\limits_{T \to \infty } \frac{1}{T}\sum_{t = 1}^{T}\sum_{s = 1}^{S_t}   r_{s,k}(t) \leq {\bar r}_k, \forall k,\tag{\ref{p1}{a}} \label{consQ_stable}\\
&0 \leq \sum_{s = 1}^{S_t}   r_{s,k}(t)  \leq {B_k},\forall k = 1,2, \ldots ,K\tag{\ref{p1}{b}}\label{consB},
\end{align}
where ${\bar r}_k$ and $B_k$ are the upper bound of average resources $k$ over time horizon $T$ and the upper bound of resources $k$, respectively.
Note that the long-term average constraint in \eqref{consQ_stable} is primarily due to the nature of resource allocation and the consideration of workload variability, fairness, and system stability.
Moreover, we consider the situations in which certain applications may require per-slot constraints to ensure strict real-time requirements.
Therefore, we also include the per-slot constraint in constraint \eqref{consB} to capture these situations.

Note that $\mathscr{P}_{\textrm{1}}$ is computationally intractable due to the non-convex objective function.
Moreover,  $\mathscr{P}_{\textrm{1}}$  contains a long-term average constraint \eqref{consQ_stable} over successive time slots.
Solving it directly requires global knowledge over time horizon $T$ about the feature values of all new users, which deviates from the purpose of online resource allocation.
To overcome these challenges, we propose to solve this problem by devising a data-driven online queue resource allocation method for homogenous and heterogeneous users.

\section{Optimal Online Queue Resource Allocation for Homogenous Users}\label{IV}
In this section, we first apply the Lyapunov optimization to transform problem ${\mathscr{P}_{\textrm{1}}}$ into a one-slot optimal problem.
Then, we propose a data-driven OOQRA algorithm to solve problem \eqref{p1}.
We give a convergence analysis for the proposed algorithm and discuss the effectiveness of the classifier.

\subsection{Lyapunov Optimization}\label{transform}
Lyapunov optimization is a drift-plus-penalty method, which has been widely adopted in the literature to decouple the long-term average constraint as per-time constraint \cite{bracciale2020lyapunov}.
Let $\tilde{\mathbf{x}}(t)=\sum_{s=1}^{S_t} \mathbf{x}_s(t) / S_t$, $\tilde{\mathbf{a}}(t)=\sum_{s=1}^{S_t} \mathbf{a}_s(t) / S_t$,  and $\tilde{\mathbf{r}}(t)=\sum_{s=1}^{S_t} \mathbf{r}_s(t) / S_t$ be the super user's feature vector, resource efficiency vector, and allocated resources at time slot $t$, respectively.
A virtual queue vector ${\mathbf Q}(t)=\left[{Q_1}(t),\ldots,{Q_K}(t)\right]^\mathsf{T}$ is introduced to control the drift of the allocated resources.
The  evolution of ${Q_k}(t+1)$, i.e., the virtual queue of the $k$-th resource of the $(t+1)$-th super user, is given by
\begin{equation}
{Q_k}(t+1) = \max \left\{ {{Q_k}(t) +{\tilde{r}_{k}}(t) - \bar{r}_k},0 \right\}, \label{updateQ}
\end{equation}
where $Q_k(0)=0, \forall k$.
Let ${{V}}(t) = \sum_{k = 1}^{K}{{{Q}}_k^2}(t) /2 $ be the Lyapunov function\footnote{A larger ${{V}}(t)$ implies a  seriously backlogged resource queue.
In other words,  a new user should wait a long time for service. },  which quantitatively characterizes the cloud server's queue states.
The one-slot conditional Lyapunov drift \cite{stochasticBook} is given by
\begin{equation}\label{cond_Lya}
\Delta {{V}}(t) \triangleq  \mathbb{E}{\rm{}}\left[ {V}(t + 1) - {V}(t) |{\mathbf{Q}} \right].
\end{equation}
Notice that the smaller $\Delta {{V}}(t)$, the more stable the system is.

Then, the drift-plus-penalty expression of the objective of the online resource allocation problem is given by
\begin{equation}
\mathbb{E}\left[\left(\Delta { V}(t)+ \theta h\left( C(t) + {\tilde{\mathbf{a}}}^\mathsf{T}(t){\tilde{\mathbf{r}}}(t) \right) \right)|\mathbf Q\right],
\end{equation}
where $\theta$ is a non-negative control parameter, striking a balance between the predicted positive probability and the consumption of cloud resources.
By analyzing the drift-plus-penalty expression, we have the following upper bound.
\begin{lemma}\label{le1}
For any control parameter $\theta \geq 0$, the drift-plus-penalty expression of  problem $\mathscr{P}_{\textrm{1}}$ is upper bounded by
\begin{equation}
\begin{aligned}
&\mathbb{E}\left[\left(\Delta { V}(t)+ \theta  h\left( C(t) + {\tilde{\mathbf{a}}}^\mathsf{T}(t){\tilde{\mathbf{r}}}(t) \right) \right)|\mathbf Q\right]\le D+\\
&\sum\limits_{k = 1}^{K}Q_k(t)\mathbb{E}\left[ \tilde{r}_{k}(t)|\mathbf Q\right]
+ \theta  \mathbb{E}\left[h\left( C(t) + {\tilde{\mathbf{a}}}^\mathsf{T}(t){\tilde{\mathbf{r}}}(t) \right)|\mathbf Q \right], \label{ddpupper}
\end{aligned}
\end{equation}
where $D = O(r^2_{\max})$ is a positive constant and $r_{\max}$ is the maximum resource allocated to the super user at time slot $t$.
\end{lemma}
\begin{proof}
See Appendix \ref{Appdx1}.
\end{proof}

Note that the upper bound can be further simplified since $D$ is a constant.
In addition, the conditional expectation can also be removed because of the principle of opportunistically minimizing an expectation in \cite{stochasticBook}.
Moreover, we omit constraint (\ref{consQ_stable}) since it can be automatically satisfied when all queues are stable by following the Lemma \ref{le2}. \begin{lemma}\label{le2}
If all virtual queues $Q_k(t)$ satisfy the stability condition, i.e., $\bar{ Q}_k =\mathop {\lim}\limits_{T \to \infty }\frac{1}{T}\sum_{t= 0}^{T} \mathbb{E}\left[Q_k(t)\right] < \infty$, the resource allocation policy will automatically guarantee the constraint (\ref{consQ_stable}).
\end{lemma}
\begin{proof}
See Appendix \ref{Appdx2}.
\end{proof}

Therefore, according to Lemmas \ref{le1} and \ref{le2}, problem ${\mathscr{P}_{\textrm{1}}}$ can be transformed into
\begin{align}
{\mathscr{P}_{\textrm{2}}}:
&\mathop {\textrm{minimize} }\limits_{\mathbf{\tilde{r}}(t)} \quad \sum\limits_{k = 0}^{K}{Q_k}(t) \tilde{r}_{k}(t) + \theta  h\left( C(t) + {\tilde{\mathbf{a}}}^\mathsf{T}(t){\tilde{\mathbf{r}}}(t) \right)\label{p2}\notag\\
&\mathop {\textrm{subject to } }\quad 0 \le {\tilde{r}_{k}}(t) \le {B_k},\forall k.
\end{align}
Note that, though $\mathscr{P}_{\textrm{2}}$ facilitates the online resource allocation design, it is still challenging to solve due to its non-convex objective function.

\subsection{Online Queue Resource Allocation (OQRA) Algorithm}\label{basisAlg}
In the following, we propose an OQRA algorithm to solve this online resource allocation problem.
The optimal resource allocation policy is required to minimize the time-average positive rate while ensuring the stability of the resource queue.
According to \cite{boyd2004convex}, the optimal solution of $\mathscr{P}_{\textrm{2}}$ needs to satisfy the KKT conditions.
In the following, we first derive the  KKT conditions of $\mathscr{P}_{\textrm{2}}$.
Then, the optimal policy can be obtained by comparing the candidate solutions, which are limited by the properties of optimal solutions from the KKT conditions.

The Lagrangian function of ${\mathscr{P}_{\textrm{2}}}$ is given by
\begin{equation}
\begin{aligned}
L\left({\tilde{\mathbf r}}(t),{\mathbf B},{\mathbf Q}(t),\boldsymbol{\tau},{\mathbf v}\right)= &{\mathbf Q}^{\mathsf{T}}(t){\tilde{\mathbf r}}(t) - {\boldsymbol{\tau} }^\mathsf{T}[{\mathbf B}-{\tilde{\mathbf r}}(t)] + {\mathbf v}^\mathsf{T}{\tilde{\mathbf r}}(t) \\ &
+ \theta  h\left( C(t) + {\tilde{\mathbf{a}}}(t){\tilde{\mathbf{r}}}^\mathsf{T}(t) \right)\label{Lagran},
\end{aligned}
\end{equation}
where ${\mathbf B} = \left[ {{B_1},{B_2},...,{B_K}} \right]^\mathsf{T}$.
The KKT conditions that the optimal solution of $t$-th super user shall satisfy are given by
\begin{align}
\theta \tilde{a}_k(t)h'\left( C(t) + {\tilde{\mathbf{a}}}^\mathsf{T}(t){\tilde{\mathbf{r}}}(t) \right) + Q_k(t) + {\tau}_k - v_k = 0&\label{kkt_dera},\\
v_k \tilde{r}_{k}(t) = 0&\label{kkt_vr},\\
{\tau}_k[B_k - \tilde{r}_{k}(t)] = 0&\label{kkt_taur},\\
v_k \ge 0&\label{kkt_v},\\
{\tau}_k \ge 0&\label{kkt_tau},\\
\forall k = 1, \dots, K&,\label{kkt_K}\notag
\end{align}
where $h'(x)$ denotes the gradient of $h(x)$.
Note that $\tilde{a}_k(t)/Q_k(t)$ turns out to be an appropriate indicator of the unit resource efficiency.
Thus, we arrange the $K$ types of resources in a non-increasing order of $p_k(t) = \tilde{a}_k(t)/Q_k(t)$, i.e.,
$p_1(t)\geq \ldots \geq p_k(t) \geq \ldots \geq p_K(t)$.
By further analyzing the KKT conditions, we have the following properties.
\newtheorem{property}{Property}
\begin{property}\label{pro1}
The $K$ types of resources are allocated sequentially, i.e., $\tilde{r}_{k{'}}(t) = B_{k{'}}$, where $\forall k{'} < k$ is a necessary condition for $\tilde{r}_{k}(t) > 0$.
\end{property}
\begin{proof} See Appendix \ref{Appdx3}. \end{proof}
\begin{property}\label{pro2}
If $\theta h'\left(C(t) + \sum_{k' = 1}^{k-1}\tilde{a}_{k{'}}(t)\tilde{r}_{k{'}}(t)\right)< -(p_k(t))^{-1}$ and  $B_k > 0$, the $k$-th resource will be allocated, i.e, $\tilde{r}_k(t)>0$.
Moreover, if $0 < Q_k(t) \le {\theta \tilde{a}_{k}(t)}/{4}$ and  then $\tilde{r}_{k}(t)$ satisfies
\begin{equation}
\tilde{r}_{k}(t) = \min \left \{\frac{\gamma(t)- C(t)-\sum_{k' = 0}^{k-1}\tilde{a}_{k{'}}(t)B_{k{'}}}{\tilde{a}_{k}(t)}, B_k \right\},\label{rk_sol1}
\end{equation}
where $\gamma(t) =\mathop{ \ln}(-1+\frac{\theta \tilde{a}_k(t)+\sqrt{\theta^{2}({\tilde{a}_k(t)})^{2}-4\theta \tilde{a}_k(t)Q_k(t)}}{2Q_k(t)})$; if $Q_k(t) = 0$, we have $\tilde{r}_{k}(t) = B_k$.
\end{property}
\begin{proof} See Appendix \ref{Appdx4}. \end{proof}

\noindent\textbf{Remark 1:} Property \ref{pro1} implies that there exists a unique order for the optimal resource allocation.
Besides, Property \ref{pro2} provides the hidden constraints that the amount of each type of resource should be allocated.

\noindent\textbf{Remark 2:} Based on Properties \ref{pro1} and \ref{pro2}, there are up to $K+1$ solutions $[{\mathbf r}^{0}(t),{\mathbf r}^{k}(t),\dots,{\mathbf r}^{K}(t)]$, where ${\mathbf r}^{k}(t)$ represents the solution in which resource $1$ to $k$ are allocated.

\begin{algorithm}[!t]
	\caption{The Online Queue Resource Allocation (OQRA) Algorithm}
	\label{OQRAA}
	\LinesNumbered
	\KwIn{${\mathbf{x}}(t)$, $\mathbf{Z}$, $\mathbf{w}_{t-1}$, ${w_{t-1,0}}$, $\theta$}
	Compute $C(t)$, and $\tilde{\mathbf a}(t)$\;
	Arrange resources as $p_1(t)\ge p_2(t)\ge \ldots \ge p_K(t)$\;
	\For{$k=1,\ldots,K$}{
			\If{$\theta  h'\left(C(t) + \sum_{k' = 1}^{k-1}\tilde{a}_{k^{'}}(t)\tilde{r}_{k^{'}}(t)\right) < -{1}/{p_k(t)}$}{
				\If{$0 < Q_k(t) \le {\theta \tilde{a}_{k}(t)}/{4}$}{
					Update $\tilde{r}_{k}(t)$ as Eq. (\ref{rk_sol1}) }
				\If{$Q_k(t)=0$}{
					Update $\tilde{r}_{k}(t)$ as $B_k$
				}								
		    }
		\If {$t\le{T}-1$}{
			Update $Q_k(t+1)$ as Eq. (\ref{updateQ})
		}
	}
	\KwOut{$\tilde{\mathbf r}(t),\mathbf Q(t+1)$}
\end{algorithm}
However, the optimal resource allocation policy can be determined by selecting the best among these solutions by given the classifier model parameters.
A similar structure can be observed in \cite{DD2017}, where the Lagrangian multiplier is the resource price.
By contrast, we regard the resource price as the resource queue length, which can quantitatively characterize the delay of the resource allocation.
Moreover,  there is no parameter to balance the objective function and resource price in \cite{DD2017};
while we strike a balance on this by adjusting the parameter $\theta$.
Based on the above analysis, we present the OQRA algorithm to solve ${\mathscr{P}_{\textrm{2}}}$ as shown in  {\bf Algorithm \ref{OQRAA}}.

\subsection{Classifier Optimization}
In the CORA framework, the feedback (i.e., the new user's feature value and label) is the foundation of the classier optimization.
%We assume that the distribution of feature values may affect user experience because of the initial overall dataset and the learned classifier model.
It not only helps us to validate the classifier model by adding a new user's feature value and label but also better quantifies the causal relationship between the feature value and user experience.
Denote  $\phi_t$ by the classifier model at time slot $t$.
The system goal is to find the optimal classifier to better guide the resource allocation, i.e.,
\begin{equation}
{\mathscr{P}_{\textrm{3}}}:
\phi^* = \mathop{\arg\min}\limits_{\phi_t \in \varphi}\mathop {\lim}\limits_{T \to \infty }\frac{1}{T}\sum\limits_{t = 1}^{T} \mathbf G\left(\mathbf x(t)+\mathbf Z(t)\mathbf r^*(\phi_t)\right),
\end{equation}\label{goal}
where $\tilde{\mathbf r}^*(\phi_t)$ is the optimal solution of ${\mathscr{P}_{\textrm{1}}}$ and  $\varphi=\{\phi_t\}^{T}_{t=1}$ is the set of classifiers.
Symbol $\mathbf G(\mathbf x)$ denotes the ground truth, representing the positive probability of the feature vector $\mathbf x$.

In fact, the optimal classifier is obtained by solving the problem ${\mathscr{P}_{\textrm{3}}}$ with known ground truth.
However, this ground truth is often difficult to obtain.
As a result, one may approximate it to true value as closely as possible.
The CORA system gradually improves the classifier by dynamically exploring the relationship between the feature values and user experience.
From a long-term perspective, we temporarily adopt a non-optimal classifier to guide the resource allocation.
Then,  the time-average positive probability can be minimized when the classifier converges or reaches the optimal point.
The process of the classifier optimization is shown in {\bf Algorithm \ref {Data-Driven Online}} (Lines 5 to 6).
In {\bf Algorithm \ref{Data-Driven Online}}, we give the details of the data-driven OOQRA algorithm.
 \begin{algorithm}[!t]
	\caption{Optimal Online Queue Resource Allocation (OOQRA) Algorithm}
	\label{Data-Driven Online}
	\LinesNumbered
	\KwIn{${\mathbf{x}_s}(t)$, $\mathbf{Z}$, $\theta$}
	    Construct the classifier model by training the initial overall dataset to obtain $w_{0,0}, \mathbf w_0$ and set ${\mathbf Q}(0)=0$\;
		\For{$t= 1,2, \ldots ,T$}{	
		    Predict the complaint rate of $t$-th super user as $u[\tilde{\mathbf{r}}(t)]$\;\label{alg_r1}
              Take this complaint rate into {\bf Alg. \ref{OQRAA}} and compute the optimal resource allocation policy $\tilde{\mathbf r}(t)$\;\label{alg_r2}
	 	    Update the overall dataset as
            \begin{equation*}
             N_t = N_{t-1}\cup\left\{\mathbf x_s(t)+\mathbf Z\mathbf r_s(t),y_s(t)\right\}_{s=1}^{S_t};
            \end{equation*}

		    Update the parameters of classifier as Eq. (\ref{weightsupdate})\; \label{alg_w2}
		}
	\KwOut{$\left[\tilde{\mathbf r}(1), \tilde{\mathbf r}(2), \dots,\tilde{\mathbf r}(T)\right]$}
\end{algorithm}

\subsection{Convergence Analysis}
At last, we provide a convergence analysis for the data-driven OOQRA algorithm.
We first analyze the performances of the classifier and resource allocation components, separately, which yields the following lemmas.
\begin{lemma} \label{lemma3}
Given a decreasing learning rate $\eta_t < {2}/{L}$, the classifier model converges to a relatively small range with a limited time, i.e.,
\begin{equation}
\sum\limits_{t = 0}^{t_1}\mathbb{E}\left[||\nabla J(\tilde {\mathbf w}_{t})||^2\right]= \epsilon_{t_1},
\end{equation}
and
\begin{equation}
\sum\limits_{t = 0}^{t_1}\mathbb{E}\left[J(\tilde {\mathbf w}_{t})-J^*\right]=\epsilon_{t_1},
\end{equation}
where  L is the Lipschitz constant of $F({ \tilde{ \mathbf w}})$ (the target function of ${\mathscr{P}_{\textrm{3}}}$).
In addition, $t_1$ is the time slot that the classifier greedily converges to a positive scale $\epsilon_{t_1}$, where $\epsilon_{t_1}< 10 ^{-3}$.
\end{lemma}
\begin{proof} See Appendix \ref{Appdx5}. \end{proof}
\begin{lemma}\label{lemma4}
Given the learned classifier model, the resource allocation component approaches the minimal complaining rate while ensuring the stability of the system queue, i.e.,
\end{lemma}
\begin{equation}
\mathop {\lim \sup }\limits_{T \to \infty }\frac{1}{T}\sum\limits_{t = t_1}^{T - 1}\sum\limits_{k = 1}^{K}\mathbb{E}\left[Q_k(t)\right] \le\frac{ D+\theta u[{{\tilde{\mathbf r}}^*}(t)]}{ B_{\max}},
\end{equation}
and
\begin{equation}
\mathop {\lim \sup }\limits_{T \to \infty } \frac{1}{T}\sum\limits_{t = t_1}^{T - 1} \mathbb{E}\left[u[{{\tilde{\mathbf r}}}(t)]\right]  \le \mathbb{E}\left[u[{{\tilde{\mathbf r}}^*}(t)]\right] +\frac{ D}{\theta},
\end{equation}
where $B_{\max}$ is the maximum of resources allocated to each super user and ${{\tilde{\mathbf r}}^*}(t)$ is the optimal resource allocation policy.
\begin{proof} See Appendix \ref{Appdx6}. \end{proof}

\noindent\textbf{Remark 3:} Lemma \ref{lemma3} demonstrates that the classifier can converge within $10^{-3}$ given a limited time.
In this case, the learned model changes slowly in which its impact on the resource allocation model is negligible.

\noindent\textbf{Remark 4:} Combining  Lemma \ref{lemma3} and Lemma \ref{lemma4},
it is reasonable to conclude that the CORA system can converge and achieve a minimum time-average complaining rate while ensuring the stability of the queue.

To sum up, we exploit the CORA framework to improve the overall user experience for homogenous users.

We successfully solve this resource allocation by applying the data-driven OOQRA algorithm based on the Lyapunov optimization technique and the optimization theory.
Theoretical results show that the proposed algorithm can converge fast.

\section{Robust Online Queue Resource Allocation for Heterogeneous Users}\label{V}
In this section, we first investigate the effect of the resource coefficient on this resource allocation problem.
Then, we propose an improved CORA framework for heterogeneous users' resource allocation.
Finally, we put forth the ROQRA algorithm for this QoE maximum problem.
 \begin{figure}[!t]
	\centering
	\includegraphics[width=0.85\columnwidth]{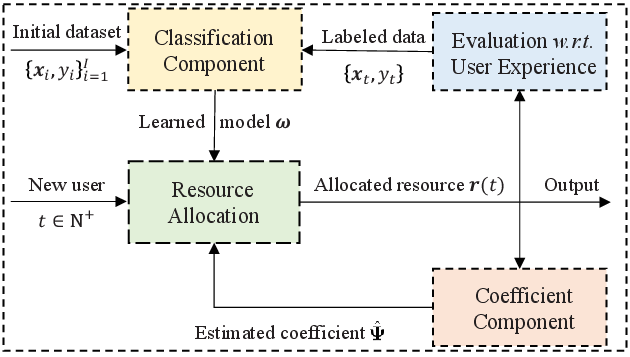}
	\caption{The improved CORA framework for the case of heterogeneous users.}
	\label{system_model}
\end{figure}

\subsection{The Effect of Resource Coefficients and the Improved CORA Framework}
In the heterogeneous user case, we assume that different users have different resource coefficient matrices $\mathbf{Z}(t)$ and only one user requests a service from the edge server at each time slot.
Thus, we omit the user index $s$ in the following.
Without loss of generality, we assume that the elements in matrix $\mathbf {Z}(t)$ follow the uniform distribution in the range of $[0, 1]$ with the same variance and different means.
In fact, the coefficient matrix $\mathbf {Z}(t)$ can be any random variable due to the user diversity.
For example, If some users may have special requirements (i.e., an extremely small $\mathbf Z(t)$),
a large part of the resources will be allocated to them to reduce their complaining rates.
In this way, other users only have fewer resources to ensure the stability of the queue in the CORA system.
This will reduce the overall user experience because of large congestion in the system.
To alleviate this, we use the estimated coefficient matrix to allocate the resource rather than the random variable $\mathbf {Z}(t)$.

Let $ \hat {\mathbf \Psi}=[\mathbf \psi_1,\ldots, \mathbf \psi_D ]^{\mathsf{T}} \in \mathbb R^{D \times K}$ be the estimated coefficient matrix at time slot $t$.
Our goal is to obtain the actual feature value $ \hat {\mathbf \Psi}$ over time slot $T$.
In other words, we expect that the feature update $\hat {\mathbf g}(t)=  \mathbf x(t)+\hat{ \mathbf \Psi}{\mathbf{r}}(t)$ not only can overcome the stochastic property of $\mathbf Z(t)$,
but also not deviate too much from the actual feature values $\mathbf g(t) = \mathbf x(t)+\mathbf Z(t){\mathbf{r}}(t)$.
%That is why we refer to the proposed method as the robust online queue resource allocation.
However, the challenge is that estimating all elements in matrix $\mathbf Z(t)$ will result in a huge communication cost due to the various resources and high dimension of user features.
To overcome this, we adopt the UCB algorithm to efficiently learn the coefficient matrix $ \hat {\mathbf \Psi}$.

Fig. \ref{system_model} illustrates an improved CORA framework for the heterogeneous user case by adding a coefficient estimation component to the original CORA framework.
Similarly, the aim of the improved CORA framework is to minimize the time average complaining rate of the heterogeneous users.
In fact, the coefficient component plays a critical role in learning the coefficient matrix by estimating from the stochastic variable $\mathbf Z(t)$.
At the beginning, the coefficient component is randomly initialized, replacing $\hat {\mathbf \Psi}(t)$ to guide resource allocation.
Then, the coefficient is further optimized in each time slot.
In this way, the problem formulation of the classifier component and resource allocation component can remain unchanged except that the coefficient matrix $\mathbf Z$ in the expression of $\mathbf{g}(t)$ is replaced by $\hat{\mathbf {\Psi}}$.

\subsection{Robust Online Queue Resource Allocation Algorithm}\label{MABAlg}
In the following, we present the ROQRA algorithm to solve problem ${\mathscr{P}_{\textrm{2}}}$  for the case of heterogeneous users, as shown in \textbf{Algorithm} \ref{FBOQRA}.
Compared with the OQRA algorithm, we model the coefficient estimation problem as an MAB problem and then adopt the UCB algorithm to learn the coefficient matrix $\hat {\mathbf \Psi}$ within the time horizon $T$.
In this MAB problem, the edge server is the player, and the elements in matrix $\mathbf{Z}$ are the arms.
At each time slot, a new user enters the MEC system and establishes a connection with the player.
Before resource allocation, the user will send the coefficient matrix $\mathbf{Z}(t)$ (i.e., the demand of resource) to the player for service.
Then, the player will observe a reward  $z_{d,k}(t)r_k(t)$, which reflects the effect of the $k$-th resource on the $d$-th feature.
Note that the allocated resource $r_k(t)$ is obtained by using the estimated coefficient $\hat{ {\psi}}_{d,k}$ in $\hat {\mathbf \Psi}$ rather than $\mathbf{Z}(t)$.
\begin{algorithm}[!t]
	\caption{Robust Online Queue Resource Allocation (ROQRA) Algorithm}
	\label{FBOQRA}
	\KwIn{}{$\theta$, $\bar {\mathbf r}$, $\mathbf z(t)$, $Q_k(0) = 0,  \bar  R_{d,k}(0) = L_{d,k}(0) = 0, \forall d, k$\;}
	Construct the classifier model by training the initial overall dataset to obtain $w_{0,0}, \mathbf w_0$ and set ${\mathbf Q}(0)= \mathbf 0$\;
	\For{$t= 1,2, \ldots ,T$}{	
		  {\bf Step 1}: Estimate the coefficient matrix $\hat {\mathbf \Psi}$ using the UCB algorithm in \eqref{ucb-index} \:

		  {\bf Step 2}: Take this coefficient matrix into \textbf{Alg.} \ref{OQRAA}

          {\bf Step 3}: Update the overall dataset and the weights using \eqref{Label_Updat01} and \eqref{Classifier_Updat02}, respectively.\:

          {\bf Step 4}: Update the estimation of $\bar  \psi_{d,k}(t) $ using \eqref{ucb-index-2} and \eqref{ucb-index-3}.
		}
\KwOut{$\left[\mathbf r(1), \mathbf r(2), \dots,\mathbf r(T)\right]$}
\end{algorithm}

From \textbf{Algorithm} \ref{FBOQRA}, we see that it first performs the UCB algorithm.
According to \cite{auer2002using}, the UCB index is computed by
\begin{equation}\label{ucb-index}
\hat \psi_{d,k}(t) = \bar  \psi_{d,k}(t-1) + c\sqrt {\frac{\log T}{L_{d,k}(t-1)}},
\end{equation}
where $c$ is a hyper-parameter that controls the exploration degree and
$\bar  \psi_{d,k}(t-1)$ denotes the average reward up to time slot $t-1$.
Symbol $L_{d,k}$ is the number of unit resources allocated by the $(d,k)$-th arm.
Note that the first term of \eqref{ucb-index} is the average empirical reward that is used for exploitation;
while the second term is an upper bound of $\bar \psi_{d,k}(t-1)$ that accounts for exploration.
Then, the estimated coefficient matrix $ \hat {\mathbf \Psi}$ is exploited to perform the resource allocation as in \textbf{Algorithm} \ref{Data-Driven Online} by using the drift-plus-penalty method.
After allocation, the testing result and  classification model are updated by
\begin{equation}\label{Label_Updat01}
N_t = N_{t-1}\cup\left\{\mathbf x(t)+\hat {\mathbf \Psi} \mathbf r(t),y(t)\right\},
\end{equation}
and
\begin{equation}\label{Classifier_Updat02}
\tilde{{\mathbf w }}_{t} = \tilde{{\mathbf w }}_{t-1} -\frac{\eta_{t}}{N_{t}}\nabla f_t(\tilde{{\mathbf w }}_{t-1}),
\end{equation}
respectively.	
Finally, the average reward $\bar  \psi_{d,k}(t)$ required in the UCB algorithm is updated by
\begin{equation}\label{ucb-index-2}
\bar  \psi_{d,k}(t) = \frac{\bar \psi_{d,k}(t-1)L_{d,k}(t-1) +z_{d,k}(t)r_k(t)}{L_{d,k}(t)},
\end{equation}
and
\begin{equation}\label{ucb-index-3}
L_{d,k}(t) = L_{d,k}(t - 1) +r_k(t).
\end{equation}
\textbf{Algorithm}  \ref{FBOQRA} performs the above steps repeatedly until it reaches the stopping time $T$.

To sum up, the ROQRA algorithm works in a dual closed-loop mode and has two main merits.
First, it overcomes the stochastic property of $\mathbf Z(t)$ through the average operation, avoiding the case where most resources are allocated to users with a small coefficient $\mathbf Z(t)$.
Second, instead of spending much time exploring all possibilities, it leverages the optimistic estimation (i.e., the UCB algorithm) to shorten the exploration time while ensuring that the valuable arms are exploited as many times as possible.

\subsection{\com{Complexity Analysis}}\label{CompAnal}
\com{At last, we provide a computational complexity analysis for the data-driven ROQRA algorithm.
We see that \textbf{Algorithm} \ref{FBOQRA} contains four steps.
Step 1 involves recalling equation \eqref{ucb-index} $DK$ times, where $D$ is the number of user features and $K$ is the number of resources. Therefore, the computational complexity of step 1 is $\mathcal{O}(TDK)$, where $T$ is the number of iterations.
Step 2 constitutes the main part of the computational complexity, as it requires calling \textbf{Algorithms} \ref{OQRAA} and \ref{Data-Driven Online} iteratively. The computational complexity of \textbf{Algorithm} \ref{OQRAA} is $\mathcal{O}(TK)$ because it involves updating the resource allocation matrix and calculating the utility function for each resource. On the other hand, the complexity of \textbf{Algorithm} \ref{Data-Driven Online} is $\mathcal{O}(T(K + DK))$, which includes predicting the complaint rate, selecting the optimal allocation, and updating the classifier model parameters.
Steps 3 and 4 involve updating the dataset and UCB index, respectively. The complexity of these steps is $\mathcal{O}(T)$.

Considering the initial dataset $I$, the overall computational complexity of the ROQRA algorithm can be expressed as
$(\mathcal{O}(I) + \mathcal{O}(TDK) + \mathcal{O}(T(K + DK)) + \mathcal{O}(T))$, which is equivalent to
$\mathcal{O}(I + TDK)$.
Since $I$ is a constant that is much smaller than $TDK$, the overall computational complexity of the ROQRA algorithm can be simplified to $\mathcal{O}(TDK)$.
In summary, the computational complexity of the data-driven ROQRA algorithm is dominated by steps 1 and 2, which involve updating the coefficient matrix and performing the optimal resource allocation. The overall complexity is $\mathcal{O}(TDK)$, indicating that the algorithm scales linearly with the number of iterations, features, and resources.}

\section{Numerical Results}\label{VI}
We conduct extensive experiments to evaluate the proposed algorithms on the synthetic and YouTube datasets.
The synthetic dataset is obtained using the Gaussian mixture distribution,  which allows us to study the performance of the proposed algorithms based on the  known ground truth;
while the YouTube dataset is used to validate the performance of the proposed algorithms.
All results are from $10^3$  independent Monte Carlo trials.

We adopt the $\varepsilon$-Perturbed DualHet algorithm \cite{DD2017} as our performance baseline.
Although this algorithm is designed for offline resource allocation, it can still be applied in our setting by considering that there is only one user in each iteration.
As suggested in \cite{DD2017} , the parameter $\varepsilon$ is chosen to be  $\varepsilon= \{0.4/t, 0.4/\mathop{\log}(t+1), 1 \}$ for comparison.
Other parameters are set in the same way as the proposed algorithms.

\subsection{The Case of Homogenous Users}
We first evaluate the OOQRA algorithm on the synthetic dataset for the case of homogenous users.
As the ground truth of the classification boundary is known,
we expect that the classifier can be optimized to approximate this ground truth.
Then, we exploit the YouTube dataset to validate the OOQRA algorithm.

\subsubsection{2D Gaussian Data}\label{Guassian}
The instances in the synthetic dataset are uniformly generated from the two-dimensional Gaussian distribution with mean $(10,10)$ and $(-10,-10)$,  and the same covariance $[20,\ 0; 0,\ 20]$.
The total number of instances is  $6,600$, where $600$ instances are regarded as the  initial  dataset;
while the remaining $6,000$ instances are assumed to arrive in sequence for service.
Assume that there is only one type of resource (e.g., the bandwidth) to allocate to the $6,000$ instances. Then, it will influence the second dimensional of user feature $x_2$, representing the change in the user's throughput.
To realize this, we set the resource coefficient matrix $\mathbf{Z}$ as $[0, 0; 0, 1]$, i.e., one resource unit increases the feature value by one unit.
In addition, assume that the resources allocated to each instance are bounded by $r_2(t)\le20-x_2(t)$.

The probability obtained by the Gaussian density function is viewed as the ground truth.
In addition, we consider the approximated ground truth to show the effect of classifier optimization,  where an instance with feature vector $\mathbf x$ is positive if $x_1\le 0$  and $x_2\le 0$.
In the following, we compare the performance of the OOQRA algorithm under the known ground truth and the approximated ground truth.
\begin{figure*}[!t]
\centering
\subfloat[$t=0$]{\includegraphics[width=2in]{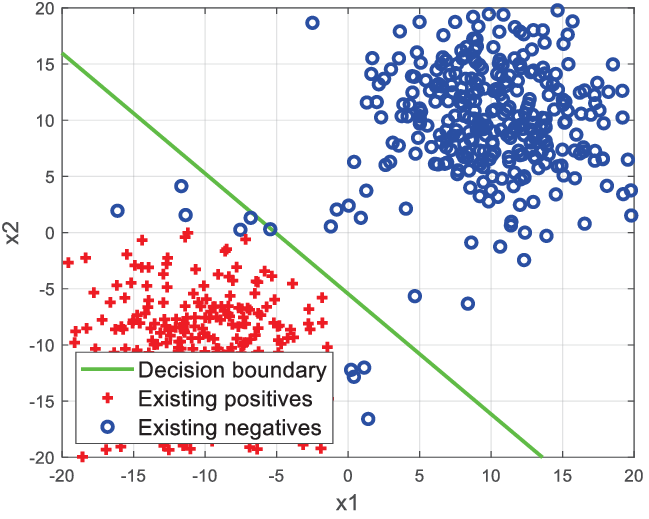}%
\label{evolution1}}
\hfil
\subfloat[$t=1,000$]{\includegraphics[width=2in]{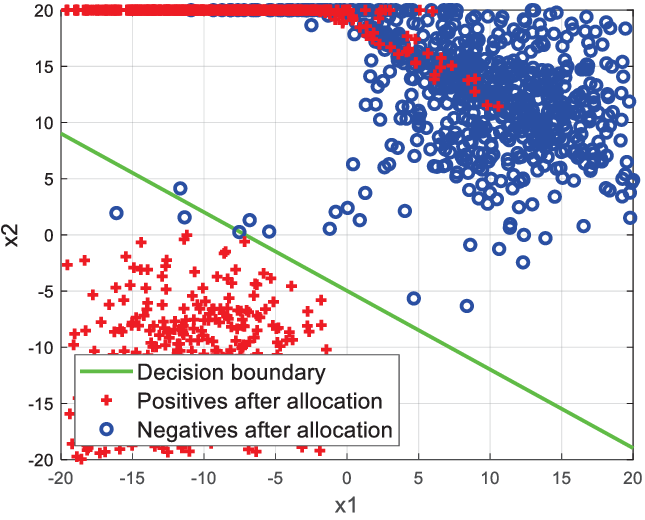}%
\label{evolution2}}
\hfil
\subfloat[$t=6,000$]{\includegraphics[width=2in]{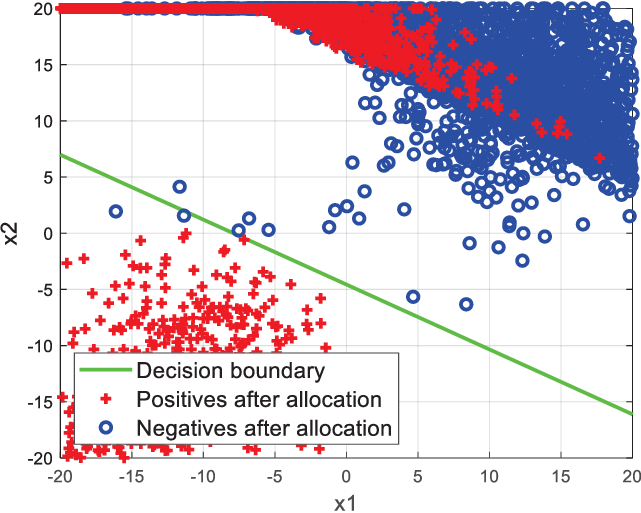}%
\label{evolution4}}
\caption{The evolution of the decision boundary for the classification component by running the OOQRA algorithm under $\bar r_2 =15$.}
\label{evolution}
\end{figure*}
Fig. \ref{evolution} shows the process of the decision boundary for the classifier component by running the OOQRA algorithm in the synthetic dataset where $\bar r_2 =15$.
The original positive and negative instances are marked as ``+'' and ``o'', respectively.
We see from  Fig. \ref{evolution1}  that the logistic regression model works efficiently under the initial $6,000$ dataset.
From  Figs.  \ref{evolution2} and  \ref{evolution4}, it is seen that some original positive instances have been moved across the decision line,
where some complaint users change to negative after resource allocation.
Fig. \ref{cvg} shows the convergence of the time average positive rate and queue length, as well as the arithmetic square root of the classier weights for the OOQRA algorithm.
We see that they trend to stable when a suitable classifier model is found.
This demonstrates that the OOQRA algorithm can converge and guarantee the stability of the virtual queue.
 \begin{figure}[t]
	\centering
	\includegraphics[width=0.8\columnwidth]{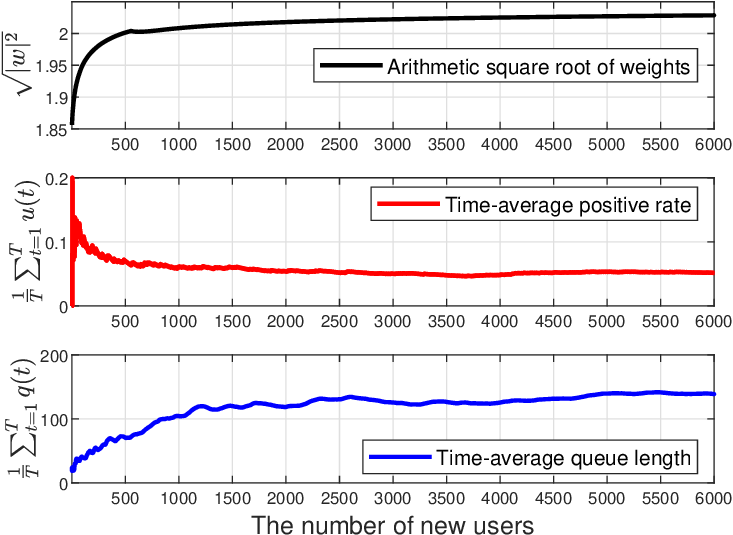}
	\caption{The convergence of the  OOQRA algorithm where $\bar r_2 =15$.}
	\label{cvg}
\end{figure}

Next, we study the evolution of time-average positive rate and queue length with different settings of the drift-plus-penalty control parameter $\theta$ in the OOQRA algorithm.
Fig. \ref{MQP}  depicts that the time average positive value gradually decreases and stabilizes as $\theta$ increases.
In contrast, the time-average queue length or delay increases quasi-linearly as $\theta$ increases,
which validates that there is a trade-off between time-average positive rate and delay as discussed in \ref{transform}.
Therefore, it is reasonable to carefully design the parameter $\theta$ to ensure the system stability for the OOQRA algorithm.
\begin{figure}[!t]
	\centering
	\includegraphics[width=0.8\columnwidth]{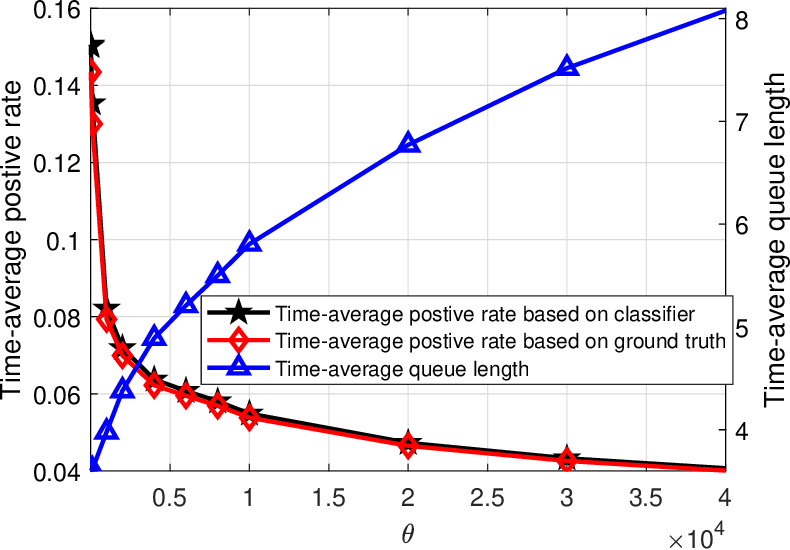}
	\caption{The drift-plus-penalty control parameter $\theta$ vs Time-average positive rate and Time-average queue length by running the OOQRA algorithm where $\bar r_2=10$.}
	\label{MQP}
\end{figure}

Finally, we show the performance of OOQRA and baseline algorithms with different settings.
We see from Fig. \ref{Resource_vs_pos} that the baseline algorithms with different probabilities require $2$-$3$ times more resources than the OOQRA algorithm to achieve the same performance.
The OOQRA algorithms with ground truth and with classifier can  reduce the time-average positive rate from $40\%$ to $10\%$ when $\bar{r}_2 = 10$,
while the baseline algorithms only reduce the time-average positive rate from $40\%$ to $30\%$.
Therefore, Fig. \ref{Resource_vs_pos}  demonstrates that the OOQRA algorithm can efficiently handle the online resource allocation.
\begin{figure}[!t]
	\centering
	\includegraphics[width=0.8\columnwidth]{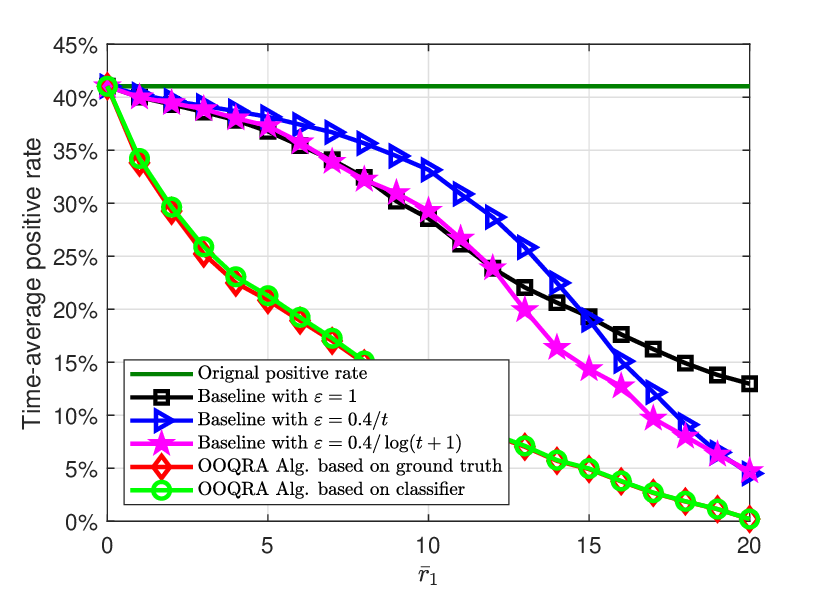}
	\caption{The performance of OOQRA and baseline algorithms with different settings in the synthetic dataset.}
	\label{Resource_vs_pos}
\end{figure}

\subsubsection{YouTube Video Data}
The causal relationship between feature values and user experience can be established from reality by properly conducting controlled experiments.
In the following, we evaluate the proposed algorithm with the YouTube video dataset in \cite{youscore}.
This dataset is used to study the interaction between the content of the YouTube video, network QoS, and QoE.
The videos are selected according to video bitstream to distinguish different video contents.
According to \cite{youscore}, the video bitstream is a quasi-linear function of the required bandwidth (i.e., resource).
Specifically,  the required bandwidth yields a smooth payout growth when the video bitstream increases.

There are a total of $2,269$ videos in the YouTube dataset.
We assume that the $2,269$ users are watching the videos, and each user watches one video.
We randomly select the feature values and QoE labels of $269$ users to construct the initial overall dataset.
The remaining $2,000$ users are assumed to arrive one by one.
The allocated resource is the bandwidth, where one unit of bandwidth can increase one unit of $x_2$.

In the following, we investigate the ground truth relationship between the network QoE and the allocated resource.
The network QoE contains the features of the round-trip time (RTT) (i.e., $x_1$) and the downlink bandwidth (i.e., $x_2$).
The ranges of the RTT and bandwidth of each video are $[40, 1000]$ ms and [0, 10] Mbps, respectively.
The QoE label is positive if the video is considered ``unacceptable''  and negative if the video is considered ``acceptable''.
Fig. \ref{RTT} shows the relationship between the RTT and positive rate.
It is seen that the positive rate increases with the RTT.
By using the curve fitting method, we can approximately estimate this relationship by
\begin{equation}
p({x_1}) =  -1.1x^2_1+2.3x_1-0.2.
\end{equation}
\begin{figure}[!t]
\centering
\subfloat[RTT v.s.  positive rate] {\label{RTT}
\includegraphics[width=0.45 \columnwidth]{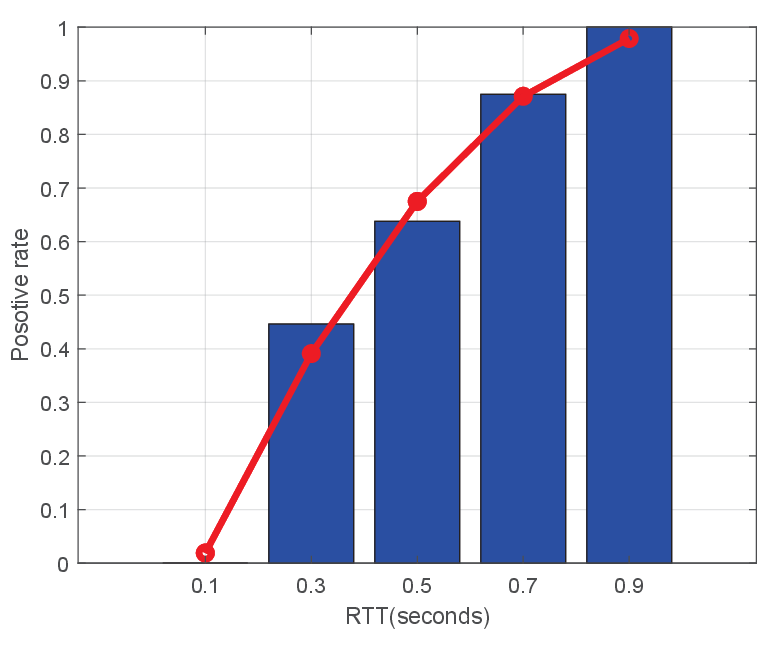}}   
\hfil
\subfloat[The ground truth] { \label{groundTruth3}   
\includegraphics[width=0.45 \columnwidth]{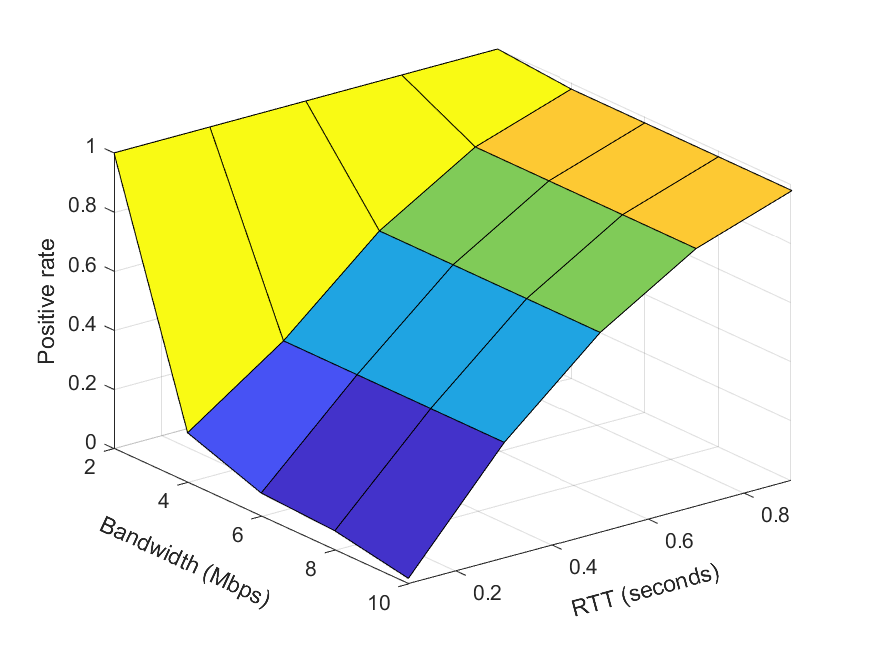}}   
\caption{(a) The relationship between  RTT and positive rate in the YouTube dataset, where the solid line is estimated from the real dataset using the curve fitting method. (b) The ground truth in the YouTube dataset.}  
\label{GroundTruthYouT}   
\end{figure}
Fig. \ref{groundTruth3}  shows the combined effect of RTT and bandwidth on users' positive rates.
It is seen that the videos are unaccepted with a probability of $1$ when the bandwidth is less than $4$ Mbps.
In addition, the acceptance of the videos depends on the RTT when the bandwidth is over $4$ Mbps.
Therefore, we have the following ground truth, i.e.,
\begin{equation}
G(\mathbf x) = \mathop {\max}(-1.1x^2_1+2.3x_1-0.2,\mathbf{1}(x_2<4)).
\end{equation}

\begin{figure}[!t]
\centering
\subfloat[Convergence analysis] {\label{pratical_r1_1M2}
\includegraphics[width=0.45 \columnwidth]{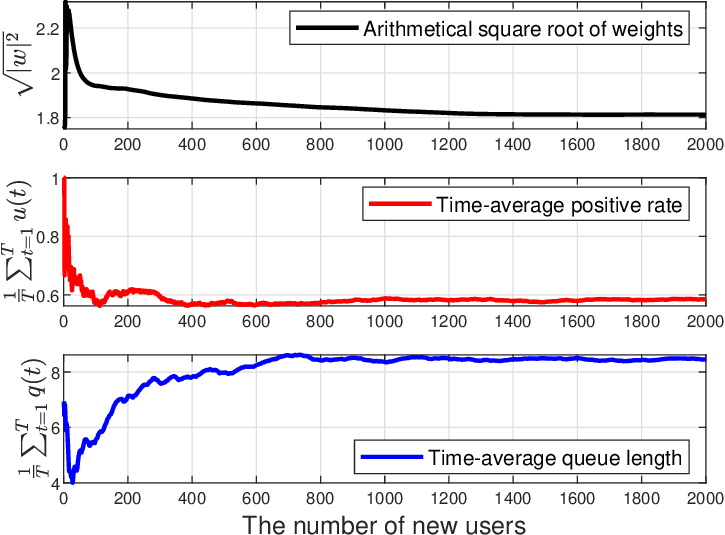}}   
\hfil
\subfloat[Performance compare] { \label{pratical_perfomance}   
\includegraphics[width=0.45 \columnwidth]{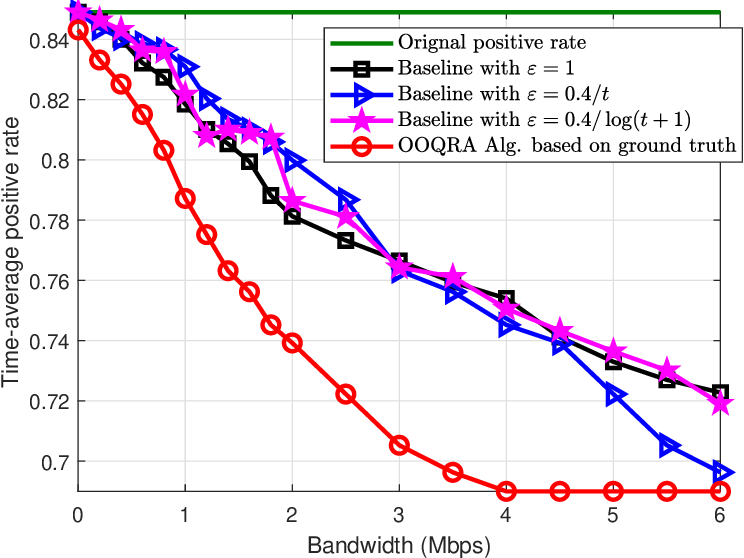}}   
\caption{(a) The convergence of the CORA system in the YouTube dataset, where $\bar r_1 = 1$ Mbps. (b) The performance of the OOQRA algorithm and the baseline algorithm in the YouTube dataset.}  
\label{OOQRA_YouTube}   
\end{figure}

Fig. \ref{pratical_r1_1M2} shows the convergence of the CORA system by running the OOQRA algorithm where $\bar r_1 = 1$ Mbps.
We see that the CORA system gradually converges to a stable state.
Meanwhile, Fig. \ref{pratical_perfomance} compares the performance of the OOQRA algorithm and the baseline algorithm.
It is seen that the complaining rate is reduced by up to $16\%$ by allocating additional bandwidth to users.
Moreover, we see that the resource efficiency is relatively high when the resource is small;
while it decreases when the resource becomes abundant.
Therefore, Fig. \ref{pratical_perfomance} demonstrates that it is possible for cloud operators to improve the overall user experience with a small price (low resource consumption).

\subsection{For the Case of Heterogeneous Users}
We investigate the effect of the coefficient component in the improved CORA framework on the synthetic dataset\footnote{Note that measuring the coefficient matrix $\mathbf Z$ in realistic is costly. }.
We assume that $\mathbf Z(t)$ is uniformly and independently distributed in $[0, 1]$.
The upper bound of the resources allocated to each user is replaced by $r_2 \le (20-x_2)/z_2$.

\begin{figure}[!t]
\centering
\subfloat[Convergence analysis] {\label{fcvg}
\includegraphics[width=0.45 \columnwidth]{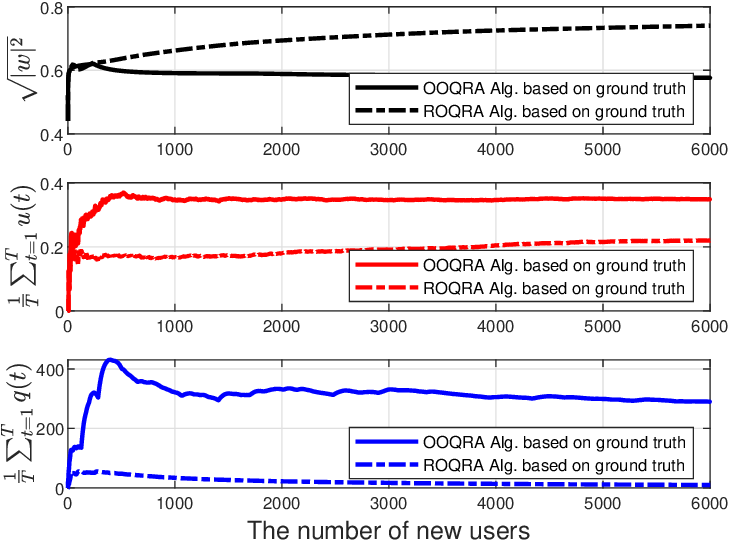}}   
\hfil
\subfloat[Performance compare] { \label{pond}   
\includegraphics[width=0.45 \columnwidth]{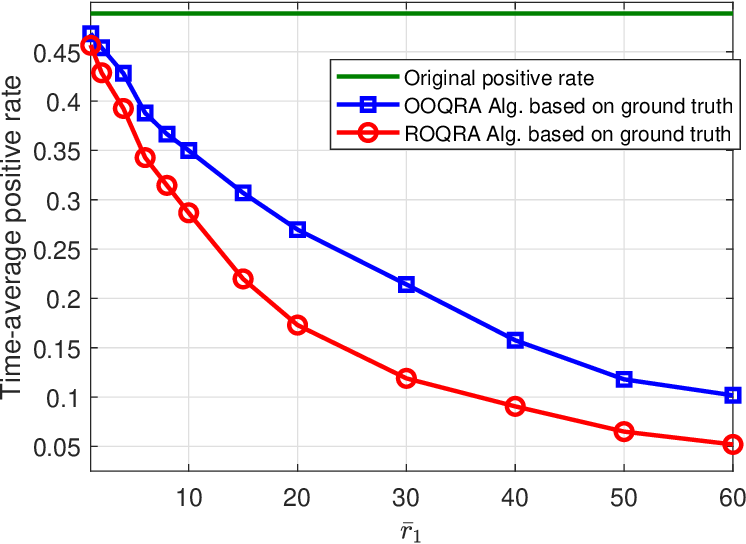}}   
\caption{(a) Convergence of the feedback-based system the synthetic dataset. (b) The performance of the OOQRA algorithm and ROQRA algorithm in the synthetic dataset.}  
\label{ROQRA}   
\end{figure}
The convergence of the improved CORA system by running the ROQRA algorithm is shown in Fig. \ref{fcvg}.
It is seen that both algorithms reach a steady state as time increases.
In addition, the ROQRA algorithm achieves a lower time-average complaining rate and queue congestion.
In contrast, the OOQRA algorithm experiences twice congestion in the resource queue compared to that of the ROQRA algorithm.
Therefore, Fig. \ref{fcvg}  implies that the complaining rate can be further reduced by using the ROQRA algorithm.
The same result is observed in Fig. \ref{pond}, where the  ROQRA algorithm only requires half of the resources to reduce the same time-average positive rate as the OOQRA algorithm.
Besides, the ROQRA algorithm only consumes $20$ units of resource on average when the time-averaged positive rate is reduced to $18\%$, while the OOQRA algorithm requires about $40$ units on average.

\section{Conclusions and Future Works}\label{VII}
We studied the  QoE-based resource allocation problem in MEC systems while considering user diversity, limited resource, and the complex relationship between the allocated resource and user experience.
We formulated this problem as an online learning problem where the new users arrive sequentially.
Then, we proposed two data-driven algorithms to solve this online learning problem for homogenous users and heterogeneous users.
Specifically, for the case of homogenous users, we propose an OOQRA algorithm based on learned model by using the Lyapunov optimization technique and the prime-dual method.
For the case of heterogeneous users, we first provided the coefficient component to learn the resource coefficient matrix $\mathbf{Z}(t)$ by using the UCB algorithm.
Then, we proposed the ROQRA algorithm to solve this problem.
Numerical results show that the user complaints rate can be reduced by up to $100 \%$ and $18 \%$  in the synthesis and YouTube datasets, respectively.

\com{In the system model, we considered that the user experience is classified as negative or positive in the dataset.
However, it is promising to consider a multivariate classification problem instead of the binary one by considering a more complex feedback system (e.g., 1-10 grading system) and applying advanced multivariate classification methods to solve this problem.}
In addition, we assumed that only one user requests service at each time slot for the heterogeneous user case.
An interesting problem is to allow multiple users to request service at the same time as that in the homogenous user case.
Meanwhile, the convergence analysis of the proposed algorithms in the multiple users scenario needs further investigation.
These are promising yet challenging problems for future study.

\appendices
\section{proof of lemma 1} \label{Appdx1}
For convenience, we consider that only one user requests service at each time slot, i.e., the super user $S_t=1, \forall t=1,2,\ldots,T$.
Thus, we have $\mathbf x (t) = \tilde{\mathbf x} (t)$, $\mathbf a (t) = \tilde{\mathbf a} (t)$, and $\mathbf r (t) = \tilde{\mathbf r} (t)$.
According to Eqs. \eqref{updateQ} and \eqref{cond_Lya}, we have
\begin{equation*}\small
\begin{aligned}
\Delta { V}(t) & =  \mathbb{E}\left[ {V}(t + 1) - {V}(t) |{\mathbf{Q}}\right]
= \frac{1}{2}\sum\limits_{k = 1}^{K} \mathbb{E}\left[{Q_k}^2(t + 1) - {Q_k}^2(t) |{\mathbf Q} \right]\\
&= \frac{1}{2}\sum\limits_{k = 1}^{K} \mathbb{E}\left[ {\left. {{{\left( \hspace{-1pt}{\max \left\{ {Q_k(t) + r(t) \hspace{-2pt} - \hspace{-2pt} {\bar r}_k,0} \right\}} \right)}^2} - {{Q_k}^2}(t)} \right|{\mathbf Q}} \right]\\
& \buildrel (a)  \over \le  \frac{1}{2}\sum\limits_{k = 1}^{K} \mathbb{E}\left[ {\left. {{{\left( {Q_k(t) + r_k(t) - \ {\bar r}_k} \right)}^2} - {{Q_k}^2}(t)} \right|{\mathbf Q}}  \right]\\
&=  \sum\limits_{k = 1}^{K}Q_k(t) \mathbb{E}\left\{{r_k(t)|{\mathbf Q}(t)} \right\} - \sum\limits_{k = 1}^{K}Q_k(t) \mathbb{E}\left[{\bar r_k(t)|{\mathbf Q}} \right]\\
&\quad+\frac{1}{2}\sum\limits_{k = 1}^{K} \mathbb{E}\left[ {\left. {{{\left( {r_k(t) - {\bar r}_k} \right)}^2}} \right|{\mathbf Q}} \right]\\
& \buildrel (b) \over\le  \sum\limits_{k = 1}^{K}Q_k(t)\mathbb{E}\left[r_k(t)|\mathbf Q\right] + D.
\end{aligned}
\end{equation*}
Inequality (a) holds since ${({\max \{ {X,0} \}})^2}\le{X^2}$.
In constraint (\ref{consQ_stable}), the term $\sum_{k = 1}^K\mathbb{E}\left[\left( r_k(t) - {\bar r}_k \right)^2|\mathbf Q\right]/2$ is upper bounded by $O( r^2_{\max})$, which is denoted by $D$.
Thus, inequality (b) can be obtained by removing the term $\sum_{k = 1}^{K}Q_k(t) \mathbb{E}\left[{\bar r}_k(t)|{\mathbf Q}\right]$  because it is positive.
By adding $\theta \cdot\mathbb{E}\left[h\left( C(t) + {\mathbf{a}}^\mathsf{T}(t){\mathbf{r}}(t) \right)|\mathbf Q\right]$ to both sides,  we can obtain  Lemma 1.

\section{proof of lemma 2} \label{Appdx2}
Similarly, we have $\mathbf x (t) = \tilde{\mathbf x} (t)$, $\mathbf a (t) = \tilde{\mathbf a} (t)$, and $\mathbf r (t) = \tilde{\mathbf r} (t)$ as only one super user request service at each time slot.
From Eq. \eqref{updateQ}, we have
\begin{equation*}
Q_k(t+1) \ge Q_k(t)+r_k-{\bar r}_k.\label{Q_update_transform}
\end{equation*}
Summing the above equation over $t=0,1, \ldots, T-1$, we have
\begin{equation*}
Q_k(T)-Q_k(0) \ge \sum\limits_{t = 0}^{T-1}r_k(t)- T{\bar r}_k.
\end{equation*}
By dividing both side by $T$ and taking a limit operator, we obtain
\begin{equation*}
\mathop {\lim}\limits_{T \to \infty }\frac{Q_k(T)-Q_k(0)}{T}\ge \mathop {\lim}\limits_{T \to \infty }\frac{1}{T}\sum\limits_{t = 0}^{T-1}r_k(t)-{\bar r}_k.\label{lemm1proof}
\end{equation*}
If the virtual queue is stable, we have $\mathop {\lim}\limits_{T \to \infty }\left({Q_k(T)-Q_k(0)}\right)/{T}=0$. Then, by rearranging the above equation, we can observe that constraint (\ref{consQ_stable}) is automatically satisfied.

\section{proof of property 1}\label{Appdx3}
Similarly, we have $\mathbf x (t) = \tilde{\mathbf x} (t)$, $\mathbf a (t) = \tilde{\mathbf a} (t)$, and $\mathbf r (t) = \tilde{\mathbf r} (t)$ as only one super user request service at each time slot.
Assume that the optimal resource allocation policy ${\mathbf r}(t)$ satisfies
\begin{equation}\label{pro_proof_asump}
-\frac{1}{p_{k+1}(t)} < \theta  h'\left( C(t) + {\mathbf{a}}(t)^\mathsf{T}{\mathbf{r}}(t) \right)\le-\frac{1}{p_k(t)}.
\end{equation}
Then, for $k{'} = 1, 2, \ldots, k-1$, ${\tau}_{k{'}}$ has to be positive to satisfy Eqs. \eqref{kkt_dera} and \eqref{pro_proof_asump}.
In the other words,  $r_{k{'}}$ must equal to $B_{k{'}}$ given Eq. \eqref{kkt_taur}. Meanwhile, for  $k{'} = k+1, k+2, \dots, K$, ${v}_{k{'}}$ has to be positive to satisfy Eqs. \eqref{kkt_dera} and \eqref{pro_proof_asump}.
So $r_{k{'}}$ must be equal to zero by given Eq. \eqref{kkt_vr}.

\section{proof of property 2}\label{Appdx4}
Similarly, we have $\mathbf x (t) = \tilde{\mathbf x} (t)$, $\mathbf a (t) = \tilde{\mathbf a} (t)$, and $\mathbf r (t) = \tilde{\mathbf r} (t)$ as only one super user request service at each time slot.
In the following, we omit the index $s$. If $\theta  h'\left(C(t) + \sum_{k' = 1}^{k-1}a_{k{'}}(t)r_{k{'}}(t)\right)< -{1}/{p_k(t)}$ and $B_k > 0$, then $r_k(t)$ has to be positive given Eqs. \eqref{kkt_dera}-\eqref{kkt_tau}; otherwise ${\tau}_{k}$ will be zero and Eq. \eqref{kkt_dera} is not satisfied.
Moreover, according to Eqs. \eqref{kkt_vr}-\eqref{kkt_tau},  $r_k(t)$ satisfies
\begin{equation}
Q_k(t)+ \theta a_k(t) h'\left( C(t) + {\mathbf{a}}(t)^\mathsf{T}{\mathbf{r}}(t) \right)= 0,\label{euql}
\end{equation}
when $0 < r_k(t)< B_k$, i.e., $\tau_k = v_k =0$.
 Then,  we have
\begin{equation*}
\begin{aligned}
&h'\left( C(t) + {\mathbf{a}}(t)^\mathsf{T}{\mathbf{r}}(t) \right)\\
&= \frac{- \mathop{\exp}\left(C(t)+\sum_{k' = 0}^{k-1}a_{k{'}}(t)B_{k{'}}+a_k(t)r_k(t)\right)}{\left(1+ \mathop{\exp}\left(C(t)+\sum_{k' = 0}^{k-1}a_{k{'}}(t)B_{k{'}}+a_k(t)r_k(t)\right)\right)^2}.
\end{aligned}
\end{equation*}
Let $s = \mathop{\exp}\left(C(t)+\sum_{k' = 0}^{k-1}a_{k{'}}(t)B_{k{'}}+a_k(t)r_k(t)\right)$. Then,  Eq. \eqref{euql} can be rewritten as
\begin{equation}
\beta(s) \triangleq  Q_k(t)s^2 + (2Q_k(t)- Ma_k(t))s + Q_k(t) = 0\label{Hfunc}.
\end{equation}
If $Q_k(t) > {\theta a_k(t)}/{4}$, there is no solution to Eq. (\ref{Hfunc}) and $\beta(s) \ge 0$.
Hence, the minimum solution of the objective function $Q_k(t)r_k(t)+ \theta h\left( C(t) + {\mathbf{a}}(t)^\mathsf{T}{\mathbf{r}}(t) \right)$ appears at $r_k(t) = 0$.
If $Q_k(t) \le {\theta a_k(t)}/{4}$, Eq. (\ref{Hfunc}) has two solutions, i.e.,
\begin{equation*}
s1, s2 = -1 + \frac{\theta a_k(t)}{2Q_k(t)} \pm \frac{\sqrt{\theta^{2}{a_k(t)}^{2}-4\theta a_k(t)Q_k(t)}}{2Q_k(t)}.
\end{equation*}
From $s1$ and $s2$, we can obtain two $r_k(t)$, which are denoted by $r_k^{1}(t)$ and $r_k^{2}(t)$, respectively.
By calculating the gradient of $\beta(s)$, we can observe that $r_k^{1}(t)$ is local maximum and $r_k^{2}(t)$ is local minimum.
Therefore, $r_k^{2}(t)$ is the optimal solution which ensures that the objective function $Q_k(t)r_k(t)+ \theta \cdot h\left( C(t) + {\mathbf{a}}(t)^\mathsf{T}{\mathbf{r}}(t) \right)$ is minimized when $Q_k(t) \le {\theta a_k(t)}/{4}$.
If $Q_k(t)=0$ and $r_k(t)>0$, $v_k$ has to be zero given Eq. \eqref{kkt_vr}.
Moreover, Eq. \eqref{kkt_dera} can be simplified to $\theta a_k(t)h'\left( C(t) + {\mathbf{a}}(t)^\mathsf{T}{\mathbf{r}}(t) \right)+ {\tau}_k = 0$, i.e., ${\tau}_k=-\theta a_k(t) h'\left( C(t) + {\mathbf{a}}(t)^\mathsf{T}{\mathbf{r}}(t) \right)\neq 0$.
Therefore, to satisfy Eq. \eqref{kkt_taur}, $r_k(t)$ has to be $B_k$ when $Q_k(t)=0$.

\section{proof of lemma 3}\label{Appdx5}
Similarly, we have $\mathbf x (t) = \tilde{\mathbf x} (t)$, $\mathbf a (t) = \tilde{\mathbf a} (t)$, and $\mathbf r (t) = \tilde{\mathbf r} (t)$ as only one super user request service at each time slot.
To prove this lemma, we first show three properties of the objective function $F(\tilde{\mathbf w})$ in ${\mathscr{P}_{\textrm{3}}}$, i.e., strongly convex, $L$-smooth, and gradient bounded.
According to the definitions, we have
\begin{align}
&F(\tilde{\mathbf w }) = - \mathop {\lim}\limits_{T \to \infty }\frac{1}{N_T}\sum\limits_{n = 1}^{N_T}\left\{y_nh({ \mathbf{x}}_n)+(1-y_n)(1-h( {\mathbf{x}}_n))\right\}, \notag\\
&\nabla F(\tilde{{\mathbf w }}) =\mathop {\lim}\limits_{T \to \infty }\frac{1}{N_T} \sum\limits_{n = 1}^{N_T}\left\{\left(h({\mathbf{x}}_n )-y_n\right)\tilde{\mathbf{x}}_n \right\},\notag\\
&\nabla^2 \left[F(\tilde{\mathbf w })\right]=\mathop {\lim}\limits_{T \to \infty }\frac{1}{N_T} \sum\limits_{n = 1}^{N_T}\left\{\tilde{\mathbf{x}}^T_n  \tilde{\mathbf{x}}_n h({\mathbf{x}}_n) (1-h({\mathbf{x}}_n)) \right\}\succeq 0\label{Hessian},
\end{align}
where $\nabla F(\tilde{{\mathbf w }})$ and $\nabla^2 [F(\tilde{{\mathbf w }})]$ are the gradient and the Hessian matrix of $F(\tilde{{\mathbf w }})$, respectively.
Let $\tilde{\mathbf{x}}_n = [1, {\mathbf{x}}_n]$, we have
\begin{equation*}
 \left[{\mathbf x}_n,y_n\right]=\left\{
\begin{array}{rll}
&\left[{\mathbf x}_i,y_i \right],            & \mathrm{if} \  { n \leq I};\\
& \left[{\mathbf g}(t),y(t)\right],          & \mathrm{otherwise}.\\
\end{array} \right.
\end{equation*}
It is obvious that $F(\tilde{{\mathbf w }})$ is strongly convex according to Eq. (\ref{Hessian}). Moreover, $F(\tilde{{\mathbf w }})$ is Lipschitz continuous gradient with Lipschitz constant $L$ based on the fact that $h({\mathbf{x}}_n) (1-h({\mathbf{x}}_n))\leq {1}/{4}$, and $L$ is smaller than the biggest eigenvalue of $\nabla ^2 [F(\tilde{{\mathbf w }})]$, i.e,
\begin{equation}
L \le \sigma_{\max}\left( \mathop {\lim}\limits_{T \to \infty }\frac{1}{4N_T} \tilde{\mathbf x}_{N_T}^T \tilde{\mathbf x}_{N_T}\right),
\end{equation}
where $\tilde{\mathbf x}_{N_T}$ is a vector and $\sigma_{\max}(A)$ represents the largest eigenvalue of a symmetric positive semidefinite matrix $A$. Since user feature values are usually limited, and the resources allocated to the user by the resource allocation model have an upper bound and are larger than $0$, $\tilde{\mathbf x}_n$ is limited. Without loss generality, we assume that $|\tilde{\mathbf x}_{n}| \le \tilde { x}_{\max}$. It can be further inferred that $\nabla F(\tilde{{\mathbf w }})$ is also upper bounded by $\tilde { x}_{\max}$. Moreover, $L$ can be rewritten as
\begin{equation*}
\begin{aligned}
L \le \sigma_{\max}\left( \frac{1}{4}\tilde{ x}^T_{\max}  \tilde{ x}_{\max}\right).
\end{aligned}
\end{equation*}
Based on the above three properties, we further show the convergence of the classifier model.
In this paper, we exploit the average gradient descent method with a decreasing step size to update the weights as follows
\begin{equation*}
\tilde{\mathbf w }_{t+1}=\tilde{\mathbf w }_{t}-\eta_t\left[ \frac{1}{N_t}\sum\limits_{n = 1}^{N_t} \nabla J_n(\tilde{{\mathbf w }}_t)\right],
\end{equation*}
where
\begin{equation*}
J_n(\tilde{\mathbf w }_t) = -\left\{y_nh_t({\mathbf{x}}_n)+(1-y_n)(1-h_t({\mathbf{x}}_n))\right\},\\
\end{equation*}
and
\begin{equation*}
h_t(x) = \frac{1}{1 + \exp( -\tilde{\mathbf w }^T_t x)}.
\end{equation*}
Thus, the gradient of $J_n(\tilde{\mathbf w }_t)$ is given by
\begin{equation}
\nabla J_n(\tilde{{\mathbf w }}_t) =\left[h_t({\mathbf{x}}_n )-y_n\right]\tilde{\mathbf{x}}_n.\\
\end{equation}
\indent As a further step, we define $\mathbf e^t_n = \nabla J_n(\tilde {\mathbf w}_{t})-\nabla F(\tilde {\mathbf w}_{t})$ and $\mathbf e^t=\sum_{n = 1}^{ N_t} e^t_n$. Then, $\tilde {\mathbf w}_{t+1}-\tilde {\mathbf w}_{t} = -\eta_t\left[ \nabla F(\tilde {\mathbf w}_{t})+{\mathbf e^t}/{ N_t}\right]$ and
\begin{equation*}\small
\begin{aligned}
&F(\tilde {\mathbf w}_{t+1})\buildrel  (\mathrm{a})\over\le F(\tilde {\mathbf w}_{t})+   \langle \nabla F(\tilde {\mathbf w}_{t}),\tilde {\mathbf w}_{t+1} - \tilde {\mathbf w}_{t}\rangle + \frac{L}{2}||\tilde {\mathbf w}_{t+1} - \tilde {\mathbf w}_{t}||^2\\
&=F(\tilde {\mathbf w}_{t})-\eta_t\langle \nabla F(\tilde {\mathbf w}_{t}),\nabla F(\tilde {\mathbf w}_{t})+\frac{\mathbf e^t}{ N_t}\rangle\\
&\quad + \frac{\eta^2_tL}{2}||\nabla F(\tilde {\mathbf w}_{t})+\frac{\mathbf e^t}{ N_t}||^2\\
&=F(\tilde {\mathbf w}_{t})-\frac{\eta_t}{ N_t}\langle \nabla F(\tilde {\mathbf w}_{t}),\mathbf e^t \rangle - \eta_t||\nabla F(\tilde {\mathbf w}_{t})||^2\\
&\quad +\frac{\eta^2_tL}{2}\left(||\frac{\mathbf e^t}{ N_t}||^2+2\frac{1}{ N_t}\langle \nabla F(\tilde {\mathbf w}_{t}),\mathbf e^t\rangle+||\nabla F(\tilde {\mathbf w}_{t})||^2\right),\\
\end{aligned}
\end{equation*}
where inequality (a) holds because the  $L$ is smooth.
With some transformations, the above equation can be rewritten  as
\begin{equation*}\small
\begin{aligned}
&\left(\eta_t-\frac{\eta^2_tL}{2}\right)||\nabla F(\tilde {\mathbf w}_{t})||^2 \notag\\
&\le F(\tilde {\mathbf w}_{t})-F(\tilde {\mathbf w}_{t+1})+\frac{\eta^2_tL-\eta_t}{N_t}\langle \nabla F(\tilde {\mathbf w}_{t}),\mathbf e^t\rangle\notag
+\frac{\eta^2_tL}{2 N^2_t}\left(||\mathbf e^t||^2\right)\notag\\
&\buildrel  (\mathrm{c})\over \le F(\tilde {\mathbf w}_{t})-F(\tilde {\mathbf w}_{t+1})+\frac{\eta^2_tL}{2 N^2_t}\left(||\mathbf e^t||^2\right)+\frac{\eta^2_tL}{N_t}|\langle \nabla F(\tilde {\mathbf w}_{t}),\mathbf e^t\rangle|\notag\\
&\quad +\frac{\eta_t}{N_t}|\langle \nabla F(\tilde {\mathbf w}_{t}),\mathbf e^t\rangle|\notag\\
&\buildrel  (\mathrm{d})\over \le F(\tilde {\mathbf w}_{t})-F(\tilde {\mathbf w}_{t+1})+\frac{\eta^2_tL}{2 N^2_t}\left(||\mathbf e^t||^2\right)\notag
+\frac{\eta^2_tL+\eta_t}{N_t}|\langle \tilde{x}_{\max},\mathbf e^t\rangle|,\notag
\end{aligned}
\end{equation*}
where inequality (c) is obtained by using  $|(a-b)c|\le |ac|+|bc|$;
while inequality (d) holds due to the upper bound of $\nabla [F(\tilde{{\mathbf w }})]$.
We assume that the variance of the stochastic gradient is bounded, i.e., $E_n\left[||\nabla f_n(\tilde {\mathbf w}_{t})-\nabla F(\tilde {\mathbf w}_{t})||^2\right]\le \sigma^2$ or $|e^t_n| < \sigma$, where $\sigma > 0$.
Thus, $E\left(||\mathbf e^t||^2\right)=E\left[||\sum\limits_{n = 1}^{ N_t}e^t_n||^2\right]\le E\left(||N_t e^t_n||^2\right)=  N^2_t\sigma^2$ due to $E\left[||X+Y||^2\right] \le E\left[\left(|X|+|Y|\right)^2\right]$.
Taking expectation on both sides with respect to the $\{N_t\}$, we have
\begin{align}
\left(\eta_t-\frac{\eta^2_tL}{2}\right)\mathbb{E}\left[||\nabla F(\tilde {\mathbf w}_{t})||^2\right]\notag
\le &\mathbb{E}\left[F(\tilde {\mathbf w}_{t})-F(\tilde {\mathbf w}_{t+1})\right] \\
&+\frac{\eta^2_tL\sigma^2}{2N^2_t} +\frac{\eta^2_tL+\eta_t}{N_t} \tilde{ x}_{\max}\sigma.\notag
\end{align}
Summing over $t =0,1,\dots,T$ and dividing both sides by $\sum_{t = 0}^{T}\left(\eta_t-{\eta^2_tL}/{2}\right)$, we obtain
\begin{equation}
\begin{aligned}
&\sum\limits_{t = 0}^{T}\mathbb{E}\left[||\nabla F(\tilde {\mathbf w}_{t})||^2\right] \leq \\
&\frac{2\mathbb{E}\left[F(\tilde {\mathbf w}_{0})-F(\tilde {\mathbf w}_{T})\right]+\sum_{t = 0}^{T}\frac{\eta^2_tL\sigma^2}{N^2_t}+ \sum_{t = 0}^{T}\frac{\eta^2_tL+\eta_t}{N_t} \tilde{ x}_{\max}\sigma}{\sum_{t = 0}^{T}\left(2\eta_t-\eta^2_tL\right) },\label{gradient_cvg}
\end{aligned}
\end{equation}
and
\begin{equation*}
\begin{aligned}
&||\tilde {\mathbf w}_{t+1}-\tilde {\mathbf w}^{*}||^2=||\tilde {\mathbf w}_{t}-\eta_t\frac{1}{N_t}\sum\limits_{n = 1}^{N_t}\nabla f_n(\tilde {\mathbf w}_{t})-\tilde {\mathbf w}^{*}||^2\\
 =&||\tilde {\mathbf w}_{t}-\tilde {\mathbf w}^{*}||^2-2\eta_t \langle \nabla F(\tilde {\mathbf w}_{t})+\frac{\mathbf e^t}{N_t},\tilde {\mathbf w}_{t}-\tilde {\mathbf w}^{*}\rangle \\
&+\eta^2_t||\nabla F(\tilde {\mathbf w}_{t})+\frac{\mathbf e^t}{N_t}||^2\\
 =&||\tilde {\mathbf w}_{t}-\tilde {\mathbf w}^{*}||^2-2\eta_t \langle \nabla F(\tilde {\mathbf w}_{t})+\frac{\mathbf e^t}{N_t},\tilde {\mathbf w}_{t}-\tilde {\mathbf w}^{*}\rangle \\
&+\eta^2_t\left(||\nabla F(\tilde {\mathbf w}_{t})||^2+2\langle \nabla F(\tilde {\mathbf w}_{t}),\frac{\mathbf e^t}{N_t}\rangle+||\frac{\mathbf e^t}{N_t}||^2\right)\\
\buildrel  (1)\over\le& ||\tilde {\mathbf w}_{t}-\tilde {\mathbf w}^{*}||^2-2\eta_t \langle \nabla F(\tilde {\mathbf w}_{t})+\frac{\mathbf e^t}{N_t},\tilde {\mathbf w}_{t}-\tilde {\mathbf w}^{*}\rangle\\
&+\eta^2_t\left(L\langle \nabla F(\tilde {\mathbf w}_{t}),\tilde {\mathbf w}_{t}-\tilde {\mathbf w}^{*}\rangle+2\langle \nabla F(\tilde {\mathbf w}_{t}),\frac{\mathbf e^t}{N_t}\rangle+||\frac{\mathbf e^t}{N_t}||^2\right)\\
\buildrel  (2)\over\le& ||\tilde {\mathbf w}_{t}-\tilde {\mathbf w}^{*}||^2-(2\eta_t-{L\eta^2_t})(F(\tilde {\mathbf w}_{t})-F^*)+\eta^2_t||\frac{\mathbf e^t}{N_t}||^2\\
&-2\eta_t \langle \tilde {\mathbf w}_{t}-\eta_t \nabla F(\tilde {\mathbf w}_{t})-\tilde {\mathbf w}^*,\frac{\mathbf e^t}{N_t} \rangle,
\end{aligned}
\end{equation*}
where $\tilde {\mathbf w}^*$ is the optimal weight of classifier.
We use the implication of  $L$-smooth in inequalities (1) and (2),
and further leverage the implication of convex and the fact of $\eta_t < {2}/{L}$.
With $||\tilde {\mathbf w}_{t}-\tilde {\mathbf w}^*|| \ge 0$, $|\nabla F(\tilde {\mathbf w}_{t})| \le \tilde{x}_{\max}$, and  $|\mathbf e^t| \le \sigma$, we have
\begin{equation*}
\begin{aligned}
&(2\eta_t-L\eta^2_t)\mathbb{E}\left[F(\tilde {\mathbf w}_{t})-F^*\right]\le \mathbb{E}\left[||\tilde {\mathbf w}_{t}-\tilde {\mathbf w}^{*}||^2\right]+\frac{\eta_t^2}{N_t}\sigma^2\\
&-\mathbb{E}\left[||\tilde {\mathbf w}_{t+1}-\tilde {\mathbf w}^{*}||^2\right] + \frac{2\eta^2_t\tilde{ x}_{\max}\sigma}{N_t}.
\end{aligned}
\end{equation*}
Summing over $t=0,1,\dots, T$ and dividing both sides of $\sum_{t = 0}^{T}(2\eta_t-L\eta^2_t)$,  it yields
\begin{equation}\label{loss_cvg}\small
\begin{aligned}
&\sum\limits_{t = 0}^{T}\mathbb{E}\left[F(\tilde {\mathbf w}_{t})-F^*\right]\\
&\le \frac{\mathbb{E}\left[||\tilde {\mathbf w}_{0}-\tilde {\mathbf w}^{*}||^2\right]-\mathbb{E}\left[||\tilde {\mathbf w}_{T}-\tilde {\mathbf w}^{*}||^2\right]+\sum_{t = 0}^{T}\frac{\eta^2_t}{N_t}\sigma(\sigma+2\tilde{ x}_{\max})}{\sum_{t = 0}^{T}(2\eta_t-L\eta^2_t)}\\
&\le \frac{\mathbb{E}\left[||\tilde {\mathbf w}_{0}-\tilde {\mathbf w}^{*}||^2\right]+\sum_{t = 0}^{T}\frac{\eta^2_t}{N_t}\sigma(\sigma+2\tilde{ x}_{\max})}{\sum_{t = 0}^{T}(2\eta_t-L\eta^2_t)},
\end{aligned}
\end{equation}
where $F^* = F(\tilde {\mathbf w}^{*})$.

Combine \eqref{gradient_cvg} with \eqref{loss_cvg},  we obtain
\begin{equation}\label{ee12}
\mathop {\lim}\limits_{T \to \infty }\sum_{t = 0}^{T}\mathbb{E}\left[||\nabla F(\tilde {\mathbf w}_{t})||^2\right]
= 0,
\end{equation}
 and
 \begin{equation}\label{ee13}
   \mathop  {\lim} \limits_{T \to \infty }\sum_{t = 0}^{T}\mathbb{E}\left[F(\tilde {\mathbf w}_{t})-F^*\right]=0,
 \end{equation}
 where
\begin{equation*}\small
\begin{aligned}
&\sum_{t = 0}^{T}\mathbb{E}\left[||\nabla F(\tilde {\mathbf w}_{t})||^2\right]
= \frac{2D_F+{L\sigma^2}\sum_{t = 0}^{T}\frac{\eta^2_t}{N^2_t}+ \frac{\eta^2_tL+\eta_t}{N_t} \tilde{\mathbf x}_{\max}\sigma}{\sum_{t = 0}^{T}\eta_t(2-\eta_t L)}\\
&\buildrel (1)\over\le \frac{2D_F+L\sigma^2\sum_{t = 0}^{T}\frac{\eta_t}{N^2_t}+(L+1)\tilde{\mathbf x}_{\max}\sigma\sum_{t = 0}^{T}\frac{\eta_t}{N_t}}{\sum_{t = 0}^{T}\eta_t}\\
&=\frac{2D_F+{L\sigma^2\sum_{t = 0}^{T}\frac{1}{(t+1)^2(I+t)}+(L+1)\tilde{\mathbf x}_{\max}\sigma}\sum_{t = 0}^{T}\frac{1}{(t+1)(I+t)}}{\sum_{t = 0}^{T}\eta_t}\\
&\le\frac{2D_F+{(L\sigma^2+(L+1)\tilde{\mathbf x}_{\max}\sigma)}\sum_{t = 0}^{T}\frac{1}{(t+1)^2}}{\sum_{t = 0}^{T}\eta_t}\\
&=\frac{2D_F+{(L\sigma^2+(L+1)\tilde{\mathbf x}_{\max}\sigma)}\frac{\pi^2}{6}}{\sum_{t = 0}^{T}\eta_t},
\end{aligned}
\end{equation*}
where in (1), we use the fact that $\eta_t^2 < \eta_t$ since $\eta_t < 1$ and $D_F= \mathbb{E}\left[F(\tilde {\mathbf w}_{0})-F(\tilde {\mathbf w}_{T})\right]$. Taking a limit of both sides, we have $\mathop {\lim} \limits_{T \to \infty }\sum_{t = 0}^{T}\mathbb{E}\left[||\nabla F(\tilde {\mathbf w}_{t})||^2\right]=0$ because $\mathop {\lim}\limits_{T \to \infty }\sum_{t = 0}^{T}\eta_t=\infty$. By denoting $\mathbb{E}\left[||\tilde {\mathbf w}_{0}-\tilde {\mathbf w}^{*}||^2\right]= D_X$, we also have $\mathop  {\lim} \limits_{T \to \infty }\sum_{t = 0}^{T}\mathbb{E}\left[F(\tilde {\mathbf w}_{t})-F^*\right]=0$.

To sum up, in the case of infinite time,
the classifier model can converge by using the average gradient descent algorithm.
Therefore,  we can infer that $\sum_{t = 0}^{t_1}\mathbb{E}\left[||\nabla F(\tilde {\mathbf w}_{t})||^2\right]$ and $\sum_{t = 0}^{t_1}\mathbb{E}\left[||\nabla F(\tilde {\mathbf w}_{t})||^2\right]$ can converge to the range of $\epsilon_t \le 10 ^{-3}$  by choosing appropriate parameters, after $t_1$ time slots.
This concludes the proof.

\section{proof of lemma 4}\label{Appdx6}
Define $\Lambda_\varepsilon = \{{\mathbf{r}}(t)|[{\mathbf{r}}(t)+\varepsilon]\in \Lambda\}, \forall t$, where $\Lambda \in [{\mathbf 0}, {\mathbf B}]$, and let $\widetilde{{\mathbf r}}(t)$ be the optimal solution of the following problem
\begin{equation}
\begin{aligned}
&\mathop {\textrm{minimize}}\limits_{{\mathbf{r}}(t) \in \Lambda_\varepsilon}\quad u[{\mathbf{r}}(t)].  \\
&\textrm{subject\ to}{\quad}\mathop {\lim}\limits_{T \to \infty } \frac{1}{T}\sum\limits_{t = 1}^{T}r_k(t) \le {\bar r}_k,\ \ \forall k.\\
\end{aligned}
\end{equation}
Note that, $u\left[{\mathbf{r}}(t)\right]$ is non-increasing. If $C(t)\ge 0$, it is convex.
If $C(t)< 0$ and  ${\mathbf{r}}(t) \in \left[-{C(t)}/{ {\mathbf{a}}(t)},{\mathbf B}\right]$, it is convex, and when $C(t)< 0$ and  ${\mathbf{r}}(t) \in [0,-{C(t)}/{ {\mathbf{a}}(t)}]$ it is concave.\\
\indent Since $\widetilde{{\mathbf r}}(t) \in \Lambda_\varepsilon$, we can deduce that there is a scheduling algorithm $\hat{{\mathbf{r}}}(t) \in \Lambda$ satisfying
\begin{equation}
\mathbb{E}\left[({\hat r_k(t)-{\bar r}_k)|{\mathbf Q}} \right] \le -\varepsilon,
\end{equation}
and
\begin{equation}\
\mathbb{E}\left[u[\hat{ \mathbf{r}}(t)]  |{\mathbf Q}\right] \le u[\widetilde{{\mathbf r}}(t)].
\end{equation}
Inspired by the Lyapunov optimization theory,
 the goal of {\bf Algorithm \ref{OQRAA}} is to minimize the upper bound of the ``drift-plus-penalty'' expression. Therefore, according to Lemma {\ref{le1}} and the above analysis, we obtain
\begin{equation}
\begin{aligned}
&\Delta { V}(t) +\theta  \mathbb{E}\left[u\left[\hat{\mathbf{r}}(t)\right]\right] \le D+\theta  u[\widetilde{{\mathbf r}}(t)]-\varepsilon\sum\limits_{k = 1}^{K}Q_k(t).
\end{aligned}
\end{equation}
Taking the expectation on both sides, it yields
\begin{equation}
\begin{aligned}
&\mathbb{E}\left[{ V}(t+1)\right] - \mathbb{E}\left[{ V}(t)\right] + \theta \mathbb{E}\left[u\left[\hat{\mathbf{r}}(t)\right]\right]\\
&\le - \varepsilon\sum\limits_{k = 1}^{K} \mathbb{E}\left[Q_k(t)\right] + \theta  \mathbb{E}\left[ u[\widetilde{{\mathbf r}}(t)]\right]+ D.
\end{aligned}
\end{equation}
Then, summing over $t = 0, 1,...,T-1$ and dividing by $T$, we have
\begin{equation}
\begin{aligned}
&\frac{1}{T}\left\{\mathbb{E}\left[{ V}(T)\right] - \mathbb{E}\left[{V}(0)\right]\right\} + \frac{1}{T}\theta \sum\limits_{t = 0}^{T - 1}\mathbb{E}\left[u[\hat{\mathbf{r}}(t)]\right]   \\
&\le D -{\varepsilon }\frac{1}{T}\sum\limits_{t = 0}^{T - 1}\sum\limits_{k = 1}^{K}\mathbb{E}\left[Q_k(t)\right]+\theta  \mathbb{E}\left[ u[\widetilde{{\mathbf r}}(t)]\right].
\end{aligned}
\end{equation}
By removing the non-negative terms, it yields
\begin{equation}
\begin{aligned}
\frac{1}{T}\sum\limits_{t = 0}^{T - 1}\sum\limits_{k = 1}^{K}\mathbb{E}\left[Q_k(t)\right] \le\frac{ D+\theta  \mathbb{E}\left[u[\widetilde{{\mathbf r}}(t)]\right]+\frac{\mathbb{E}\left[{ V}(0)\right]}{T}}{{\varepsilon }}.
\end{aligned}
\end{equation}
Taking a  limit operation and let $\varepsilon \le {B}_{\max}$, we have
\begin{equation}
\begin{aligned}
\mathop {\lim \sup }\limits_{T \to \infty } \frac{1}{T}\sum\limits_{t = 0}^{T - 1}\sum\limits_{k = 1}^{K}\mathbb{E}\left[Q_k(t)\right]  \le\frac{ D+\theta}{{{B}_{\max} }},
\end{aligned}
\end{equation}
and
\begin{equation}
\begin{aligned}
 \frac{1}{T}\sum\limits_{t = 0}^{T - 1} \mathbb{E}\left[u[\hat{\mathbf{r}}(t)]\right] \le \mathbb{E} \left[u[\widetilde{{\mathbf r}}(t)]\right] +\frac{ D}{\theta} + \frac{\mathbb{E}\left[{ V}(0)\right]}{T\theta}.
\end{aligned}
\end{equation}
Then, taking a  limit operation, we obtain
\begin{equation}
\begin{aligned}
\mathop {\lim \sup }\limits_{T \to \infty } \frac{1}{T}\sum\limits_{t = 0}^{T - 1}\mathbb{E} \left[u[\hat{\mathbf{r}}(t)]\right]   \le \mathbb{E}\left[u[\widetilde{{\mathbf r}}(t)]\right]+\frac{ D}{\theta} .
\end{aligned}
\end{equation}
Suppose $\left\{{{\mathbf r}^*}(t)\right\}^{T}_{t=1}$ is the optimal solution of ${\mathscr{P}_{\textrm{1}}}$. Note that ${{\mathbf r}^*}(t)\ge \widetilde{{\mathbf r}}(t)$, since $u[{{\mathbf r}}(t)]$ is non-increasing. By taking the limit $\varepsilon \to 0$
\begin{equation}
\begin{aligned}
\mathop {\lim \sup }\limits_{T \to \infty } \frac{1}{T}\sum\limits_{t = 0}^{T - 1}\mathbb{E}\left[u[\hat{\mathbf{r}}(t)\right]  \le \mathbb{E}\left[u[{{\mathbf r}^*}(t)]\right]+\frac{ D}{\theta} .
\end{aligned}
\end{equation}
Since $\Lambda_\varepsilon \subseteq \Lambda$, we have $u[\widetilde{{\mathbf r}}(t)] \ge u[{{\mathbf r}^*}(t)]$. Thus
\begin{equation}
\begin{aligned}
\mathop {\lim \inf }\limits_{\varepsilon \to 0}u[\widetilde{{\mathbf r}}(t)] \ge u[{{\mathbf r}^*}(t)].
\end{aligned}
\end{equation}
\indent Due to $\widetilde{{\mathbf r}}(t),{{\mathbf r}^*}(t) \in \Lambda$ and $\varepsilon \le B_{\max}$, we have $(1-{\varepsilon}/{ B_{\max}}){{\mathbf r}^*}(t)+{\varepsilon}/{ B_{\max}}\widetilde{{\mathbf r}}(t) \in \Lambda$. Note that $(1-{\varepsilon}/{ B_{\max}}){{\mathbf r}^*}(t)\in \Lambda_\varepsilon$. Thus, we have
\begin{equation}
\begin{aligned}
u[\widetilde{{\mathbf r}}(t)] \le u\left[(1-\frac{\varepsilon}{B_{\max}}){{\mathbf r}^*}(t)\right].
\end{aligned}
\end{equation}
By using the Taylor expansion, we have
\begin{equation}
\begin{aligned}
u[(1-\frac{\varepsilon}{ B_{\max}}){{\mathbf r}^*}(t)] &\approx -u^{'}[{{\mathbf r}^*}(t)]\frac{\varepsilon}{ B_{\max}}{{\mathbf r}^*}(t)+ u[{{\mathbf r}^*}(t)]\\
 &\le- u'\left[-\frac{C(t)}{ {\mathbf{a}}(t)}\right] \frac{\varepsilon}{B_{\max}}{{\mathbf r}^*}(t) + u[{{\mathbf r}^*}(t)].\\
\end{aligned}
\end{equation}
Hence,
\begin{equation}
\begin{aligned}
\mathop {\lim \inf }\limits_{\varepsilon \to 0}u[{{\mathbf r}^*}(t)] \ge u[\widetilde{{\mathbf r}}(t)].
\end{aligned}
\end{equation}
As a result,
\begin{equation}
\begin{aligned}
\mathop {\lim }\limits_{\varepsilon \to 0}u[{{\mathbf r}^*}(t)] = u[\widetilde{{\mathbf r}}(t)].
\end{aligned}
\end{equation}
This concludes the proof.

\balance

\bibliographystyle{./IEEEtran}
\bibliography{./IEEEabrv,./resources}
\end{document}